\documentclass[12pt,twoside]{article}

 \usepackage{float}
\usepackage{graphicx}
\usepackage{epstopdf}
\usepackage{graphicx}
\usepackage{epic}
\usepackage{multirow}
\usepackage{tikz}
\usepackage{threeparttable}
\usepackage{xcolor}
\usetikzlibrary{arrows,shapes,chains}

\renewcommand{\paragraph}{\roman{paragraph}}
\usepackage[a4paper]{geometry}
\makeatletter
\renewcommand\title[1]{\gdef\@title{\reset@font\Large\bfseries #1}}
\renewcommand\section{\@startsection {section}{1}{\z@}%
                                   {-3.5ex \@plus -1ex \@minus -.2ex}%
                                   {2.3ex \@plus.2ex}%
                                   {\normalfont\large\bfseries}}
\renewcommand\subsection{\@startsection{subsection}{2}{\z@}%
                                     {-3ex\@plus -1ex \@minus -.2ex}%
                                     {1.5ex \@plus .2ex}%
                                     {\normalfont\normalsize\bfseries}}
\renewcommand\subsubsection{\@startsection{subsubsection}{3}{\z@}%
                                     {-2.5ex\@plus -1ex \@minus -.2ex}%
                                     {1.5ex \@plus .2ex}%
                                     {\normalfont\normalsize\bfseries}}

\def\@runningauthor{}\newcommand{\runningauthor}[1]{\def\runningauthor{#1}}
\def\@runningtitle{}\newcommand{\runningtitle}[1]{\def\runningtitle{#1}}

\renewcommand{\ps@plain}{%
\renewcommand{\@evenhead}{\footnotesize\scshape \hfill\runningauthor\hfill}
\renewcommand{\@oddhead}{\footnotesize\scshape \hfill\runningtitle\hfill}}

\newcommand{\F}{\mathbb{F}}
\newcommand{\x}{\mathbf{x}}

\g@addto@macro\bfseries{\boldmath}

\makeatother

\usepackage{amsthm,amsmath,amssymb}
\usepackage{cite}
\usepackage{graphicx}

\usepackage[colorlinks=true,citecolor=black,linkcolor=black,urlcolor=blue]{hyperref}

\theoremstyle{plain}
\newtheorem{theorem}{Theorem}[section]

\newtheorem{lem}[theorem]{Lemma}
\newtheorem{cor}[theorem]{Corollary}
\newtheorem{prop}[theorem]{Proposition}

\theoremstyle{definition}

\newtheorem{example}[theorem]{Example}

\newtheorem{open}[theorem]{Open Problem}

\theoremstyle{remark}
\newtheorem{remark}[theorem]{Remark}

\runningauthor{}

\date{}

\begin{document}

\title{Binary self-orthogonal codes which meet the Griesmer bound or have optimal minimum distances}
\author{Minjia Shi\thanks{smjwcl.good@163.com}, Shitao Li\thanks{lishitao0216@163.com}, Tor Helleseth\thanks{tor.helleseth@uib.no}, Jon-Lark Kim\thanks{jlkim@sogang.ac.kr}
\thanks{Minjia Shi and Shitao Li are with School of Mathematical Sciences, Anhui University, Hefei, China. Tor Helleseth is with Department of Informatics, University of Bergen, Bergen, Norway.
Jon-Lark Kim is with Department of Mathematics, Sogang University, Seoul, South Korea.}}

\date{}
    \maketitle

\begin{abstract}
The purpose of this paper is two-fold. First, we characterize the existence of binary self-orthogonal codes meeting the Griesmer bound by employing Solomon-Stiffler codes and some related residual codes.
Second, using such a characterization, we determine the exact value of $d_{so}(n,7)$ except for five special cases and the exact value of $d_{so}(n,8)$ except for 41 special cases, where $d_{so}(n,k)$ denotes the largest minimum distance among all binary self-orthogonal $[n, k]$ codes. Currently, the exact value of $d_{so}(n,k)$ $(k \le 6)$ was determined by Shi et al. (2022).
In addition, we develop a general method to prove the nonexistence of some binary self-orthogonal codes by considering the residual code of a binary self-orthogonal code.
\end{abstract}

\section{Introduction}
Since the beginning of coding theory, the classification of binary self-dual or self-orthogonal (for short, SO) codes has been one of the most active research problems \cite{SO-SD}.
 There are several reasons why they become so interesting and popular.
 First, some interesting SO codes include the binary simplex code $\mathcal{S}_k$ for $k\geq 3$, the extended binary $[8,4,4]$ Hamming code $\widehat{\mathcal{H}}_3$, the extended binary and ternary Golay codes, and the first order binary Reed-Muller codes $\mathcal{R}(1,k)$ for $k\geq 3$.
  Second, SO codes have close connections with other mathematical structures such as combinatorial $t$-design theory \cite{t-design}, group theory \cite{C-1}, Euclidean or Hermitian lattice theory \cite{C-1,B-1,H-1}, and modular forms \cite{Shi-book}. More specifically, many finite groups such as the Mathieu groups appear as the groups of some SO codes. The Conway group is related to the extended binary SO $[24,12,8]$ Golay code. Many new 5-designs were found from SO codes \cite{5-design}.
  Besides their interesting algebraic and combinatorial structures, they have applications in quantum information theory and can be employed to construct quantum codes \cite{Quantum-1,Quantum-2}. However, finding SO codes with good minimum distances are non-trivial. More precisely, let $d_{so}(n,k)$ denote the largest minimum distance among all binary SO $[n,k]$ codes.
Then the determination of $d_{so}(n,k)$ has been a fundamental and difficult problem in coding theory because there are too many binary SO codes as the dimension $k$ increases.

It is well-known that SO codes form an important class of codes which are asymptotically good \cite{SO-good} and have been extensively studied over different alphabets, such as the Kleinian four group (the direct sum of $\F_2$ and $\F_2$) \cite{Annalen}.
Pless \cite{Pless-1} gave a classification of self-dual codes with even $n$ between 2 and 20, and SO $[n,\frac{n-1}{2}]$ codes with odd $n$ between 3 and 19.
In 2006, Bouyukliev et al. \cite{SO-40} completed the characterization of binary optimal SO codes for $n\leq 40$ and $k\leq 10$, and determined the exact value of $d_{so}(n,3)$. Later, Li et al. \cite{Li-Xu-Zhao} partially characterized the exact value of $d_{so}(n,4)$ by systems of linear equations. Kim et al. \cite{Kim-embedding} completely determined the remaining cases and partially characterized the exact value of $d_{so}(n,5)$ by embedding linear codes into SO codes. Recently, Kim and Choi \cite{Kim-SO} constructed many new optimal binary SO codes by considering the self-orthogonality matrix and gave two conjectures on $d_{so}(n,k)$ for $k=5$ or 6. Very recently, Shi et al. \cite{SO-5-6} solved the two conjectures proposed by Kim and Choi, and determined the exact value of $d_{so}(n,k)$ for $k=5$ or 6.
Readers can refer to \cite{SO-JCTA-1,SO-JCTA-2,SO-DDC-32,Shi-1} for the details on the classification of binary SO codes.
Furthermore, Kim and Choi \cite{Kim-SO} also proposed an open problem, namely,
\begin{open}
Find new optimal SO codes with $n\geq 30$ and dimension $k\geq 7$.
\end{open}

On the other hand, another fundamental problem in coding theory is to characterize the existence of Griesmer codes.
Constructing a linear code that meets the Griesmer upper bound has became a popular research problem.
In 1965, Solomon and Stiffler \cite{SS-code} presented a class of Griesmer codes by systematically puncturing certain coordinates of the Simplex codes. Belov \cite{Belov-1} gave a family of linear codes meeting the Griesmer bound after a slight generalization and reformulation. In \cite{Tor-Til-1} and \cite{Tor-Til-2}, Helleseth and van Tilborg constructed some linear codes meeting the Griesmer bound which could not be obtained from the Solomon and Stiffler or the Belov constructions. Later, Helleseth \cite{Tor-1983} gave a new construction of Griesmer codes, which generalized the results of Solomon and Stiffler \cite{SS-code}, Belov \cite{Belov-1}, and Helleseth and van Tilborg \cite{Tor-Til-1}, \cite{Tor-Til-2}.
Another important result is that Helleseth \cite{Tor-4} proved that any binary $[n,k]$ Griesmer code with the minimum distance at most $2^{k-1}$ is either a Solomon-Stiffler code or a Belov code. There is also a characterization of Griesmer codes using minihypers in a finite projective geometry~\cite{Ham}.

However, there has been not much attention on binary SO codes which meet the Griesmer bound, which is mainly due to the fact that there is no general method to construct SO codes meeting the Griesmer bound and that there are not many such codes. In this paper, we characterize binary SO Griesmer codes.

The main contribution of this paper is to characterize optimal binary SO codes including SO Griesmer codes.
We characterize the exact value of $d_{so}(n,k)$ by employing binary SO Griesmer codes.
Our contributions are summarized as follows.

\begin{itemize}
\setlength{\itemsep}{1.5pt}
\setlength{\parsep}{1.5pt}
\setlength{\parskip}{1.5pt}
  \item [(1)] First, we characterize binary SO Griesmer codes based on the Solomon-Stiffler codes, the binary Simplex codes, and the first order binary Reed-Muller codes.
  \item [(2)] Second, we prove that the Belov codes of non-Solomon-Stiffler type are not SO. We also present a sufficient and necessary condition for the Solomon-Stiffler codes to be SO. As a consequence, we determine the exact value of $d_{so}(n,k)$ where $n$ is large relative to $k$ (see Theorems \ref{thm-1} and \ref{thm-3}). In other words, we reduce a problem with an infinite number of cases to a finite number of cases.
  \item [(3)] By considering the residual code of a binary SO code, we develop a general method to prove the nonexistence of some binary SO codes. To be specific, we obtain the residual code of a binary SO code by combining with the self-orthogonality of the SO code in order to  determine the first few rows of the generator matrix of the SO code, and we finally get a contradiction. In addition, we completely solve the remaining case of $k=6$ in \cite{Kim-SO} and \cite{SO-5-6}. We determine the exact value of $d_{so}(n,7)$ except for five special cases and the exact value of $d_{so}(n,8)$ except for 41 special cases.
\end{itemize}

The paper is organized as follows. In Section 2, we give some notations and preliminaries. In Section 3, we present a general construction method for binary SO codes and study the conditions for the existence of binary SO Griesmer codes. In Section 4, we construct binary SO Griesmer codes from Solomon-Stiffler codes and Belov codes. In Section 5, we present an asymptotic result on the largest minimum distance of binary SO codes. In Section 6, we prove the nonexistence of some binary SO codes with dimension 7 by using the residual codes to approach it. In Section 7, we conclude the paper.

\section{Preliminaries}
Let $\F_2$ denote the finite field with $2$ elements. A binary {\em linear $[n,k]$ code} is a $k$-dimensional subspace of $\F_2^n$. For any ${\bf x}=(x_1,x_2,\ldots,x_n)\in \F_2^n$, the {\em support} of ${\bf x}$ is defined as follows:
$$supp({\bf x})=\{i~|~x_i=1\}.$$
The {\em Hamming weight} ${\rm wt}({\bf x})$ of ${\bf x}$ is the number of nonzero components of ${\bf x}$, i.e., ${\rm wt}({\bf x})=|supp({\bf x})|$.
The {\em minimum (Hamming) distance} of a linear code $C$ is defined to be the smallest nonzero Hamming weight of all codewords in $C$.
A binary linear $[n,k,d]$ code $C$ is a binary linear $[n,k]$ code with the minimum distance $d$.
The dual code $C^{\perp}$ of a binary linear code $C$ is defined as
$$C^{\perp}=\{\textbf y\in \F_2^n~|~\langle \textbf x, \textbf y\rangle=0, {\rm for\ all}\ \textbf x\in C \},$$
where $\langle \textbf x, \textbf y\rangle=\sum_{i=1}^n x_iy_i$ for $\textbf x = (x_1,x_2, \ldots, x_n)$ and $\textbf y = (y_1,y_2, \ldots, y_n)\in \F_2^n$.
A binary linear code $C$ is {\em self-orthogonal} (SO) if $C\subseteq C^\perp$. In particular, $C$ is {\em self-dual} if $C=C^\perp.$

There are many bounds on the minimum distance of linear codes, and one of them is the Griesmer bound (see \cite{Griesmer} and \cite[Chap. 2, Section 7]{Huffman}), which is defined on a binary linear $[n,k,d]$ code $C$ as
$$n\geq g(k,d)=\sum_{i=0}^{k-1}\left\lceil \frac{d}{2^i}\right\rceil,$$
 where $\lceil a \rceil$ is the least integer greater than or equal to a real number $a$.
A binary linear $[n, k, d]$ code $C$ is said to be a {\em Griesmer code} if $n$ meets the Griesmer bound, i.e., $n=g(k,d)$.
A binary linear $[n,k,d]$ code $C$ is optimal with respect to the Griesmer bound if $n\geq g(k,d)$ and $n<g(k,d+1)$.

Let $d(n,k)$ denote the largest minimum distance among all binary linear $[n,k]$ codes.
Let $d_{so}(n,k)$ denote the largest minimum distance among all binary SO $[n, k]$ codes.
A binary SO $[n,k]$ code is {\em optimal} if it has the minimum distance $d_{so}(n,k)$.
A vector $\x = (x_1, x_2, \ldots , x_n)\in \F^n_2$
is {\em even-like} if $\sum_{i=1}^nx_i=0$ and is {\em odd-like} otherwise.
An even-like vector $\x = (x_1, x_2, \ldots , x_n)\in \F^n_2$
is {\em doubly-even} if ${\rm wt}({\bf x})$ is a multiple of 4 and is {\em singly-even} otherwise.
A binary linear code is said to be {\em even-like} if it has only even-like codewords, and is said to be {\em odd-like} if it is not even-like. An even-like linear code is said to be {\em doubly-even} if it has only doubly-even codewords, and is said to be {\em singly-even} if it is not doubly-even.

\begin{remark}
Due to the self-orthogonality in SO codes, the value $d_{so}(n,k)$ is always even. In addition, the best possible minimum distance of a binary SO $[n,k]$ code is $2\left\lfloor d(n,k)/2\right\rfloor$, that is to say, $d_{so}(n,k)\leq 2\left\lfloor d(n,k)/2\right\rfloor$.
\end{remark}


Assume that $S_k$ is a matrix whose columns are all nonzero vectors in $\F_2^k$.
It is well-known that $S_k$ generates a binary simplex code, which is an one-weight SO $[2^k-1,k,2^{k-1}]$ Griesmer code for $k\geq 3$ (see \cite{Huffman}).
Consider the following matrix
$$R(1,k)=\left[\begin{array}{cc}
    1 & 1\cdots1  \\
    0 & S_k
  \end{array}\right],$$
which generates the first order Reed-Muller code $\mathcal{R}(1,k)$ (see \cite{Huffman}).
It can be checked that $\mathcal{R}(1,k)$ is a binary SO $[2^k,k+1,2^{k-1}]$ code for $k\geq 3$.

\section{Binary optimal SO codes related with the Simplex codes}

The following lemma shows that we can construct a family of SO codes from an SO code.

\begin{lem}\label{lemma-k-3}
 Let $G^*=[G_0~|~G]$, where $G_0$ generates a binary SO $[n_0,k,d_0]$ code $C_0$ and $G$ is a $k\times n$ matrix (the rank of $G$ can be less than $k$). Then, $G$ generates a binary SO code if and only if $G^*$ generates a binary SO $[n+n_0,k]$ code $C^*$.
In particular, if $G$ generates a binary $[n,k,d]$ code, then $C^*$ has the minimum distance at least $d+ d_0$.
\end{lem}

\begin{proof}
Since $C_0$ is a binary SO $[n_0,k,d_0]$ code, $G_0G_0^T=O_{k \times k}$. It turns out that
$$G^*{G^*}^T=GG^T.$$
Therefore, $G$ generates a binary SO code if and only if $C^*$ is a binary SO $[n+n_0,k]$ code. In particular, suppose that $G$ generates a binary $[n,k,d]$ code. Since $C_0$ has the minimum distance $d_0$, $C^*$ has the minimum distance at least $d+d_0$. This completes the proof.
\end{proof}

\begin{remark} \label{rem-short}
If we take some special cases, then we have the following results.
\begin{itemize}
\setlength{\itemsep}{1.5pt}
\setlength{\parsep}{1.5pt}
\setlength{\parskip}{1.5pt}
  \item [(1)] If $G_0=[\underbrace{S_k~|~\cdots~|~S_k}_m]$, then Lemma \ref{lemma-k-3} is the same as \cite[Lemma 3.1]{SO-5-6}.
  \item [(2)] Let $C$ be a binary SO $[n,k,d]$ code with $d\leq m2^{k-1}$ for some $m \ge 1$ and generator matrix $G$. If $G_0=[\underbrace{R(1,k)~|~\cdots~|~R(1,k)}_m]$, where $R(1,k)=\left[\begin{array}{cc}
    1 & 1\cdots1  \\
    0 & S_k
  \end{array}\right]$, then it can be checked that the following matrix
$$G'=\left[\begin{array}{c|c|c|c}
             R(1,k)&\cdots&R(1,k) & \begin{array}{ccc}
     0\ldots0 \\
    G
  \end{array}
           \end{array}
\right]$$
generates a binary SO $[m2^k+n,k+1,m2^{k-1}+d]$ code.
\end{itemize}
\end{remark}



\begin{lem}\label{lem-griesmer}
Assume that $S_k$ is a matrix whose columns are all nonzero vectors in $\F_2^k$ for $k\geq 3$. Let $C$ be a binary $[n,k,d]$ linear code with generator matrix $G$.
Then $C$ is an optimal binary $[n,k,d]$ linear code with respect to the Griesmer bound if and only if the following matrix
$$G'=[\underbrace{S_k~|~\cdots~|~S_k}_m|~G]$$
 generates an optimal binary $[N=m(2^k -1) +n, K=k, D=2^{k-1} m + d]$ code $C'$ with respect to the Griesmer bound for any $m \ge 0$.
In particular, $C$ is a Griesmer code if and only if $C'$ is a Griesmer code.
\end{lem}

\begin{proof}
Note that the parameters of $C'$ are justified in Lemma~\ref{lemma-k-3}.
Suppose that $C$ is optimal with respect to the Griesmer bound. Then it satisfies $n\geq \sum_{i=0}^{k-1} \left \lceil \frac{d}{2^i} \right \rceil$ and $n< \sum_{i=0}^{k-1} \left \lceil \frac{d+1}{2^i} \right \rceil$. Hence
\begin{eqnarray*}
  \sum_{i=0}^{K-1}  \left \lceil \frac{D}{2^i} \right \rceil  & = &  \sum_{i=0}^{k-1}  \left \lceil \frac{D}{2^i} \right \rceil \\
   & = & \sum_{i=0}^{k-1}  \left \lceil \frac{2^{k-1}m +d}{2^i} \right \rceil \\
   & = & (2^{k-1}m +d ) + (2^{k-2}m + \left\lceil \frac{d}{2} \right\rceil) + \cdots + (m+\left\lceil \frac{d}{2^{k-1}} \right\rceil) \\
   & = & (2^k -1)m + \sum_{i=0}^{k-1} \left \lceil \frac{d}{2^i} \right \rceil.
\end{eqnarray*}
Since $n\geq \sum_{i=0}^{k-1} \left \lceil \frac{d}{2^i} \right \rceil$ and $n< \sum_{i=0}^{k-1} \left \lceil \frac{d+1}{2^i} \right \rceil$, we have
$N\geq \sum_{i=0}^{k-1} \left \lceil \frac{D}{2^i} \right \rceil$ and $N< \sum_{i=0}^{k-1} \left \lceil \frac{D+1}{2^i} \right \rceil$.
Therefore $C'$ is an optimal binary $[N,K,D]$ code with respect to the Griesmer bound.

Conversely, if $C'$ is optimal with respect to the Griesmer bound, i.e.,
\begin{center}
$N\geq \sum_{i=0}^{k-1} \left \lceil \frac{D}{2^i} \right \rceil$ and $N< \sum_{i=0}^{k-1} \left \lceil \frac{D+1}{2^i} \right \rceil$,
 \end{center}
 then it also implies that
 $n\geq \sum_{i=0}^{k-1} \left \lceil \frac{d}{2^i} \right \rceil$ and $n< \sum_{i=0}^{k-1} \left \lceil \frac{d+1}{2^i} \right \rceil$ by a similar argument as above.
In particular, it can be seen that $n=\sum_{i=0}^{k-1} \left \lceil \frac{d}{2^i} \right \rceil$ if and only if
$N= \sum_{i=0}^{k-1} \left \lceil \frac{D}{2^i} \right \rceil$. This completes the proof.
\end{proof}

By Lemma~\ref{lemma-k-3}, Remark \ref{rem-short}, and Lemma \ref{lem-griesmer}, we have the following useful theorem for the SO codes satisfying the Griesmer bound.

\begin{theorem}\label{thm-large-griesmer}
Assume that $S_k$ is a matrix whose columns are all nonzero vectors in $\F_2^k$ for $k\geq 3$. Let $C$ be a binary $[n,k,d]$ linear code with generator matrix $G$.
Then $C$ is an optimal binary SO $[n,k,d]$ code with respect to the Griesmer bound for given $n$ and $k$ if and only if the following matrix
$$G'=[\underbrace{S_k~|~\cdots~|~S_k}_m~|~G]$$
 generates an optimal binary SO $[N=m(2^k -1) +n, K=k, D=2^{k-1} m + d]$ code with respect to the Griesmer bound for any $m \ge 0$. In particular, $C$ is a SO Griesmer code if and only if $C'$ is a SO Griesmer code.
\end{theorem}


\begin{example}
We note that there is no SO [22, 5, 10] code by~\cite[Theorem 4.3]{SO-5-6}.
However, there exists a binary SO $[53, 5, 26]$ code (see \cite{Kim-SO}). Once can explicitly check that $d=26$ is the largest minimum distance satisfying the Griesmer bound for $n=53$ and $k=5$. Therefore, by adding $\ell-1$ copies of $S_5$ to this code, we see that there exists a binary SO $[N=31 (\ell-1) + 53, K=5, D=16 (\ell-1) + 26]$ code for $\ell \ge 2$ where $D=16\ell+10$ is the largest minimum distance satisfying the Griesmer bound because $D=16 \ell +11$ implies $N=\sum_{i=0}^4 \left \lceil \frac{16 \ell+11}{2^i} \right \rceil = 31\ell +23$, which is a contradiction by the Griesmer bound.
\end{example}

\begin{remark}
Similarly, it can be checked that Theorem \ref{thm-large-griesmer} holds for (2) of Remark \ref{rem-short}.
\end{remark}

Let ${\bf c}\in C$ be a codeword of weight $\omega.$
Let $Res(C,{\bf c})$ denote the residual code of $C$ with respect to ${\bf c}$, which is the code of length $n -\omega$ punctured on all the coordinates of the support set of ${\bf c}$.
 Sometimes we write  instead $Res(C,w)$ when ${\bf c}$ has weight $w$.
The following is a useful lemma for the residual code of a binary SO code with parameters $[n,k,d]$ for $d \equiv 2 \pmod{4}$, which also shown in \cite{Li-Xu-Zhao}.
For completeness, let us prove it.

\begin{lem} \label{SO-Lemma}
Let $C$ be a binary SO code with parameters $[n,k,d]$ for $d \equiv 2 \pmod{4}$. Then the residual code of $C$ with respect to a codeword of minimum weight $d$ has parameters
$\left[n-d,k-1,\frac{d}{2}+1\right]$.
\end{lem}

\begin{proof}
In the general case, the residual code of $C$ with respect to a codeword of even weight $d$ has parameters $\left[n-d,k-1,\frac{d}{2}\right]$.
Suppose $C$ is a binary SO $[n,k,d]$ code with $d \equiv 2 \pmod{4}$. Let $C_0=Res(C,d)$ denote the residual code of $C$ with respect to a minimum codeword ${\bf c}_1 \in C$ of weight $d$. Then $C_0$ has parameters $[n-d, k-1,\frac{d}{2}]$. Let ${\bf c}_0 \in C_0$ be a codeword of weight $\frac{d}{2}$ that is the restriction of a codeword ${\bf c}_2 \in C$.
Then since both ${\bf c}_2$ and ${\bf c}_1 + {\bf c}_2$ have weight at least $d$, the inner product  $\langle {\bf c}_1 , {\bf c}_2\rangle = \frac{d}{2} \equiv 1 \pmod{2}$. Since, the code $C$ is SO this is impossible and therefore a codeword of weight $\frac{d}{2}$ in $C_0 $ is not possible, implying that $C_0$ has the minimum distance at least $\frac{d}{2}+1$.
\end{proof}



\begin{prop}\label{prop-D-double-even}
Suppose that $k\geq 3$. Then a binary $[N,k,D]$ Griesmer code is SO if and only if $C$ is a doubly-even code.
\end{prop}

\begin{proof}
Let $C$ be a binary $[N,k,D]$ SO Griesmer code for $k\geq 3$. Then $N=\sum_{i=0}^{k-1}\left\lceil \frac{D}{2^i}\right\rceil$.
It follows from self-orthogonality that $D$ is even.
Since $\left\lceil \frac{D+2}{2^0}\right\rceil=\left\lceil \frac{D}{2^0}\right\rceil+2$ and $\left\lceil \frac{D+2}{2}\right\rceil=\left\lceil \frac{D}{2}\right\rceil+1$,
 $$\sum_{i=0}^{k-1}\left\lceil \frac{D+2}{2^i}\right\rceil\geq
 \sum_{i=0}^{k-1}\left\lceil \frac{D}{2^i}\right\rceil+3=N+3.$$

If $D$ is singly-even, then we assume that $D=2D'$, where $D'$ is odd. By Lemma \ref{SO-Lemma}, there is an even-like binary  $[N-2D',k-1,D'+1]$ linear code $C'$. Then
\begin{align*}
  \sum_{i=0}^{k-2}\left\lceil \frac{D'+1}{2^i}\right\rceil & = \sum_{i=0}^{k-2}\left\lceil \frac{D+2}{2^{i+1}}\right\rceil\\
   & =\sum_{i=1}^{k-1}\left\lceil \frac{D+2}{2^{i}}\right\rceil\\
   &=\sum_{i=0}^{k-1}\left\lceil \frac{D+2}{2^{i}}\right\rceil-D-2\\
   &\geq N+3-D-2\\
   &=N-2D'+1.
\end{align*}
So the code $C'$ does not satisfy the Griesmer bound, which is a contradiction. This implies that $D$ is doubly-even.
Since $C$ is a Griesmer code, $C$ has generator matrix consisting of doubly-even minimum weights~\cite[Theorem 2.7.6]{Huffman}. Therefore, $C$ is doubly-even since $C$ is SO.

Conversely, if $C$ is a doubly-even code, then it is well known~\cite[Theorem 1.4.8 (ii)]{Huffman} that $C$ is SO. This completes the proof.
\end{proof}

\section{Binary SO Griesmer codes from Solomon-Stiffler codes}


Assume that $S_k$ is a $k\times (2^k-1)$ matrix whose columns are made up of all nonzero vectors of $\F_2^k$.
Let $[U\setminus V]$ denote a matrix whose columns consist of the elements of $U$ which do not contain the elements of $V$, where $U$ and $V$ are subspaces of $\F_2^k$.
Let $sS_k$ denote $s$ copies of a matrix $S_k$, in other words,
$$sS_k=[\underbrace{S_k~|~S_k~|~\cdots~|~S_k}_s].$$
Let $G$ is a $k\times n$ submatrix of rank $k$ of $sS_k$. Then the matrix $G$ generates a binary $[n,k,d]$ linear code $C$, that is,
$$C=\{{\bf x}G~|~{\rm x}\in \F_2^k\}.$$
Hence, there exists a $k\times (s(2^k-1)-n)$ matrix $G'$ such that
$$sS_k=[G~|~G'].$$
The rank of $G'$ can be less than $k$. Then the matrix $G'$ generates a code $C'$, which is called the anticode of $C$, namely,
$$C'=\{{\bf x}G'~|~{\bf x}\in \F_2^k\}. $$
For ${\bf c}\in C$, there exists ${\bf x}\in \F_2^k$ such that ${\bf c}={\bf x}G$. We define ${\bf c'}={\bf x}G'$, then ${\bf c'}\in C'$. Let ${\bf c_1}=({\bf c},{\bf c'})$. It turns out that
\begin{align}\label{eq-1}
  wt({\bf c_1})= & wt({\bf c})+wt({\bf c'})=s2^{k-1},~{\bf c}\in C\backslash \{{\bf 0}\}.
\end{align}

We define
\begin{center}
$\mathcal{U}(k,u)=\left\{U~|~U=\widehat{U}\backslash \{{\bf 0}\},~\widehat{U}~{\rm is~a}~u\right.$-dimensional subspace of $\left.\F_2^k\right\}.$
\end{center}

Given $k$ and $d$, we define $s=\left\lceil\frac{d}{2^{k-1}}\right\rceil$.
Let $u_1,u_2,\ldots,u_p$ be positive integers such that $k>u_1> u_2>\cdots >u_p\geq 1$
 and
 $$s2^{k-1}-d=\sum_{i=1}^p2^{u_i-1}.$$

\subsection*{Binary Solomon-Stiffler codes}

In 1965, Solomon and Stiffler~\cite{SS-code} found a family of Griesmer codes by specifying $G'$ as follows:
\begin{align*}
  G'=\left[U_1~|~U_2~|~\cdots~|~U_p\right]\subset sS_k,
\end{align*}
where $U_i\in \mathcal{U}(k,u_i)$ and $k>u_1> u_2>\cdots >u_p\geq 1$. 
Note that the matrix $G'$ generates a binary code $C'$ with the maximum weight $\sum_{i=1}^p 2^{u_i-1}$. By (\ref{eq-1}) and the Griesmer bound,
the code $C$ with the anticode $C'$ is a binary linear code with the parameters
\begin{align}\label{eq-2}
\left[s(2^k-1)-\sum_{i=1}^p(2^{u_i}-1),k,s2^{k-1}-\sum_{i=1}^p2^{u_i-1}\right].
\end{align}
Then $C$ is called a binary Solomon-Stiffler code~\cite{SS-code}.
In 1974, Belov \cite{Belov-1} showed that the Solomon-Stiffler code is a Griesmer code if $\sum_{i=1}^{\min\{s+1, p\}}u_i\leq sk$. That is to say, if $\sum_{i=1}^{\min\{s+1, p\}}u_i\leq sk$, then there is a binary Solomon-Stiffler code with the parameters (\ref{eq-2}).
Note that such a code is projective when $s=1$.

\subsection*{Binary Belov codes}
We define $\Gamma(u)$~\cite{Belov-1} such that
$$\Gamma(u)=\left\{T~|~|T|=u+1,~\sum_{{\bf t}\in T}{\bf t}={\bf 0},~{\rm Rank}(T)=u\right\}.$$
In 1974, Belov \cite{Belov-1} constructed a new family of Griesmer codes by specifying $G'$ as follows
\begin{align*}
  G'=[U_1~|~U_2~|~\cdots~|~U_t~|~S\backslash T~|~R]\subset sS_k,
\end{align*}
where $k>u_1> u_2>\cdots >u_t>u\geq 3$, $U_i\in \mathcal{U}(k,u_i)$, $S\in \mathcal{U}(k,u+1)$, $T\subseteq S$, $T\in \Gamma(u+1)$, $R\in \mathcal{U}(k,1)$ if $d$ is odd, $R=\emptyset$ if $d$ is even.

Note that when $d$ is even, $C'$ has the maximum weight
$$\sum_{i=1}^t2^{u_i-1}+2^u-2=\sum_{i=1}^t2^{u_i-1}+\sum_{j=2}^{u}2^{j-1}.$$
When $d$ is odd, $C'$ has the maximum weight
$$\sum_{i=1}^t2^{u_i-1}+2^u-1=\sum_{i=1}^t2^{u_i-1}+\sum_{j=1}^{u}2^{j-1}.$$
Hence we have
$$u_i=u+1+t-i,~{\rm if}~t+1\leq i\leq p$$
and $u_p=1$ if $d$ is odd, $u_p=2$ if $d$ is even.
Thus  $$s2^{k-1}-d=\sum_{i=1}^p2^{u_i-1}=\sum_{i=1}^t2^{u_i-1}+\sum_{i=t+1}^p2^{u_i-1}=
\left\{\begin{array}{ll}
         \sum_{i=1}^t2^{u_i-1}+2^u-1 & {\rm if}~d~{\rm is~odd},\\
         &\\
         \sum_{i=1}^t2^{u_i-1}+2^u-2 & {\rm if}~d~{\rm is~even}.
       \end{array}
\right.$$
Suppose that $d$ is odd. The code $C$ with the anticode $C'$ is a binary linear code with the parameters
\begin{align}\label{eq-3}
\left[s(2^k-1)-\sum_{i=1}^t(2^{u_i}-1)-2^{u+1}+u-2,k,s2^{k-1}-
\left(\sum_{i=1}^t2^{u_i-1}+2^u-1\right)\right].
\end{align}
Suppose that $d$ is even. The code $C$ with the anticode $C'$ is a binary linear code with the parameters
\begin{align}\label{eq-4}
\left[s(2^k-1)-\sum_{i=1}^t(2^{u_i}-1)-2^{u+1}+u-1,k,s2^{k-1}-
\left(\sum_{i=1}^t2^{u_i-1}+2^u-2\right)\right].
\end{align}
Then we call $C$ a binary code of Belov type (Belov codes)~\cite{Belov-1}.

\begin{remark}
Note that $S\backslash T=\emptyset$ if $u=1$ and $S\backslash T\in \mathcal{U}(k,2)$ if $u=2$. Thus, if we allow $u = 1$ or $u = 2$ in the definition of the Belov codes, then we can consider the Solomon-Stiffler codes as a subclass of the Belov codes. In fact, it is more
convenient to treat the two families of codes separately in many cases.
\end{remark}

Since the binary SO code is even, we only consider the Belov code with parameters (\ref{eq-4}). Now let us prove that Belov codes are not SO.

\begin{theorem}\label{thm-Belov}
We keep the above notation. Then the Belov code is not SO.
\end{theorem}

\begin{proof}
Since a binary SO code is even, we only consider the Belov code with parameters (\ref{eq-4}). Since the minimum distance of the Belov code with parameters (\ref{eq-4}) is singly-even, it follows from Proposition \ref{prop-D-double-even} that the Belov code is not SO.
\end{proof}

\begin{theorem}\label{SO-SS-codes}
Keeping the above notation. Assume that $k\geq 4$. Then the following three statements are equivalent.
\begin{itemize}
  \item [{\rm(1)}] $u_p\geq 3$.
  \item [{\rm(2)}] The binary Solomon-Stiffler codes are SO codes.
  \item [{\rm(3)}] The binary Solomon-Stiffler codes are doubly-even.
\end{itemize}
\end{theorem}

\begin{proof}
$(1)\Longrightarrow (2)$.
Assume that $G=\left[s S_{k}\setminus G'\right]$, where $G'=\left[ S_{u_1}~|~\cdots~|~S_{u_p}\right]$. Let $C$ and $C'$ be codes with the generator matrices $G$ and $G'$, respectively.
Then we have
$$s S_k=[G~|~G'],$$
which generates a binary $[s(2^k-1),k,s 2^{k-1}]$ SO Griesmer code for $k\geq 4$.
For ${\bf u,v}\in C$, there exist ${\bf x,y}\in \F_2^k$ such that ${\bf u}={\bf x}G$ and ${\bf v}={\bf y}G$. We define ${\bf u'}={\bf x}G'$ and ${\bf v'}={\bf y}G'$ so that ${\bf u',v'}\in C'$. Let ${\bf u_1}=({\bf u},{\bf u'})$ and ${\bf v_1}=({\bf v},{\bf v'})$. It turns out that
$${\rm wt}({\bf u_1})={\rm wt}({\bf u})+{\rm wt}({\bf u'})=s2^{k-1},~{\rm wt}({\bf v_1})={\rm wt}({\bf v})+{\rm wt}({\bf v'})=s2^{k-1},$$
and ${\bf u_1} \cdot {\bf v_1}={\bf u}\cdot {\bf v}+{\bf u'}\cdot {\bf v'}=0.$

If $u_p\geq 3$, then $u_i\geq 3$ for $1\leq i\leq p$ and
$$\dim({\rm rowspace}(S_{u_i}))={\rm Rank}(S_{u_i})=u_i\geq 3.$$
Hence the matrix $S_{u_i}$ generates a $[2^{u_i}-1,u_i,2^{u_i-1}]$ Simplex code, which is SO (see \cite{Huffman}).
It follows that the code $C'$ with the generator matrix $G'$ is SO, which implies that ${\bf u'}\cdot {\bf v'}=0$.
Hence for any ${\bf u}, {\bf v}\in C$,
$${\bf u}\cdot {\bf v}={\bf u_1} \cdot {\bf v_1}-{\bf u'}\cdot {\bf v'}=0,$$
 that is to say, $C$ is also SO.

$(2)\Longrightarrow (3)$. Since the binary Solomon-Stiffler codes are Griesmer codes, it follows from Proposition \ref{prop-D-double-even} that SO Solomon-Stiffler codes are doubly-even.

$(3)\Longrightarrow (1)$. Let $C$ be a Solomon-Stiffler code with the minimum distance $d=s2^{k-1}-\sum_{i=1}^p2^{u_i-1}$, where $k>u_1> u_2>\cdots >u_p\geq 1$. If $C$ is doubly-even, then $d=s2^{k-1}-\sum_{i=1}^p2^{u_i-1}$ is doubly-even. It turns out that $u_p\geq 3$. This completes the proof.
\end{proof}

\begin{cor}\label{cor-SO-SS}
There exists a binary $[s(2^k-1)-\sum_{i=1}^p (2^{u_i}-1),k,s2^{k-1}- \sum_{i=1}^p2^{u_i-1}]$ SO Griesmer code for $k\geq 4$, $u_p\geq 3$ and $\sum_{i=1}^{\min\{s+1,p\}}u_i\leq sk$.
\end{cor}

\begin{proof}
Since the proof is straightforward by Theorem \ref{SO-SS-codes}, we omit it here.
\end{proof}

\begin{cor}\label{cor-all-so}
The only binary $[n,k]$ SO Griesmer code with minimum distance at most $2^{k-1}$ is a binary Solomon-Stiffler code with parameters $[(2^k-1)-\sum_{i=1}^p (2^{u_i}-1),k,2^{k-1}- \sum_{i=1}^p2^{u_i-1}]$, where $k\geq 4$, $u_p\geq 3$ and $\sum_{i=1}^{\min\{2, p\}}u_i\leq k$.
\end{cor}

\begin{proof}
According to \cite[Theorem 1.3]{Tor-4}, any binary $[n,k]$ Griesmer code with the minimum distance at most $2^{k-1}$ is either a Solomon-Stiffler code or a Belov code. By Theorem \ref{thm-Belov}, the Below codes are not SO. Combined with Corollary \ref{cor-SO-SS}, we can obtain desired result.
\end{proof}

So far we have characterized all $[n, k, d]$ Griesmer SO codes for given $n$ and $k$ such that $d \le 2^{k-1}$. We describe them when $k=6$ as follows.

\begin{example}
We construct SO Griesmer codes of dimension $k=6$ by Corollary \ref{cor-SO-SS}.
\begin{itemize}
  \item For $k=6$ and $s=1$.
By Corollary \ref{cor-SO-SS}, there exists a binary SO Griesmer code with parameters
 $$\left[63-\sum_{i=1}^p(2^{u_i}-1),6,32-\sum_{i=1}^p2^{u_i-1}\right],$$
where $6>u_1>u_2>\cdots>u_p\geq 3$ and $\sum_{i=1}^{\min\{2,p\}}u_i\leq 6$.
\begin{center}
Table 1\\
\begin{tabular}{l|l}
  \hline
  $p=0$ & $p=1$      \\
  \hline
  $[63,6,32]$ &$[56,6,28]$ \\
  &$[48,6,24]$   \\
  &$[32,6,16]$   \\
  \hline
\end{tabular}
\end{center}
By Corollary \ref{cor-all-so}, the codes in Table 1 are all binary $[n,6]$ SO Griesmer codes with the minimum distance at most $32$.
  \item For $k=6$ and $s=2$. By Corollary \ref{cor-SO-SS}, there exists a binary SO Griesmer code with parameters
 $$\left[126-\sum_{i=1}^p(2^{u_i}-1),6,64-\sum_{i=1}^p2^{u_i-1}\right],$$
where $6>u_1>u_2>\cdots>u_p\geq 3$ and $\sum_{i=1}^{\min\{3,p\}}u_i\leq 12$.

\begin{center}Table 2\\
\begin{tabular}{l|l|l|l}
  \hline
  $p=0$ & $p=1$      &     $p=2$ &$p=3$ \\
  \hline
  $[126,6,64]$ &$[119,6,60]$ & $[104,6,52]$ &   $[73,6,36]$\\
                &$[111,6,56]$ &  $[88,6,44]$& \\
               &$[95,6,48]$ &  $[80,6,40]$ & \\
  \hline
\end{tabular}
\end{center}
\end{itemize}
\end{example}

\begin{remark}
Combining the above two tables and Theorem \ref{thm-large-griesmer}, we construct a SO Griesmer code with the minimum distance $D$ and dimension $k=6$ from SO Solomon-Stiffler codes, where $D$ is doubly-even, $D\geq 16$ and $D\neq 20$.
\end{remark}

Now we give an asymptotic result.

\begin{theorem}\label{thm-asymptotic}
Suppose that $k\geq 4$ is an integer and $D\geq m2^{k-1}$ is doubly-even. If $m\geq \min\left\{\left\lceil \frac{(k+2)(k-3)}{2k}\right\rceil,
\left\lceil \sqrt{2k+\frac{1}{4}}-\frac{3}{2}\right\rceil\right\}-1$, then there exists a binary SO $[N,k,D]$ Griesmer code.
\end{theorem}

\begin{proof}
Suppose that $D\geq m2^{k-1}$ is doubly-even. Then there exist $s,u_1,u_2,\ldots,u_p$ such that $s\geq m+1$, $k>u_1>u_2\cdots>u_p\geq 1$ and $D=s2^{k-1}-\sum_{i=1}^p2^{u_i-1}$. Since $D$ is doubly-even, $u_p\geq 3$. When $m\geq \min\left\{\left\lceil \frac{(k+2)(k-3)}{2k}\right\rceil,~
\left\lceil \sqrt{2k+\frac{1}{4}}-\frac{3}{2}\right\rceil\right\}-1$, that is, $s\geq \min\left\{\left\lceil \frac{(k+2)(k-3)}{2k}\right\rceil,~
\left\lceil \sqrt{2k+\frac{1}{4}}-\frac{3}{2}\right\rceil\right\}$,
it can be checked that
\begin{align*}
  \sum_{i=1}^{\min\{s+1,p\}}u_i & \leq \min\left\{\sum_{i=1}^{s+1}u_i,~\sum_{i=1}^{p}u_i\right\} \\
  & \leq \min\left\{\sum_{i=k-(s+1)}^{k-1}i,~\sum_{i=3}^{k-1}i\right\}\\
  &=\min\left\{\frac{(2k-s-2)(s+1)}{2},~\frac{(k+2)(k-3)}{2}\right\}\\
  &\leq  sk.
\end{align*}
According to \cite{Belov-1}, there is a binary Solomon-Stiffler code with the parameters
$$\left[N=s(2^k-1)-\sum_{i=1}^p(2^{u_i}-1),k,D=s2^{k-1}-\sum_{i=1}^p2^{u_i-1} \right],$$
which is a SO Griesmer code by Theorem \ref{SO-SS-codes}. This completes the proof.
\end{proof}

\begin{remark}\label{remark-1111}
Note that
$$\min\left\{\left\lceil \frac{(k+2)(k-3)}{2k}\right\rceil,
\left\lceil \sqrt{2k+\frac{1}{4}}-\frac{3}{2}\right\rceil\right\}=
\left\{\begin{array}{ll}\vspace{0.2cm}
         1,& k=4,\\\vspace{0.2cm}
         2, & k=5,\\
         \left\lceil \sqrt{2k+\frac{1}{4}}-\frac{3}{2}\right\rceil, & k\geq 6.
       \end{array}
\right.$$
\end{remark}

\section{Optimal binary SO codes}

In this section, we completely characterize optimal binary SO codes when $n$ is large relative to $k$. Recall that $g(k,d)=\sum_{i=0}^{k-1}\left\lceil \frac{d}{2^i}\right\rceil$.

\begin{lem}\label{lem-g}
If there is a binary $[g(k,d),k,d]$ SO Griesmer code, then
$$d_{so}(N,k)=d,$$
where $g(k,d)\leq N\leq g(k,d+2)$.
\end{lem}

\begin{proof}
By Proposition \ref{prop-D-double-even}, $d$ is doubly-even.
Let $N_1=g(k,d)$ and $N_2= g(k,d+2)$.
Suppose that $N_1\leq N\leq N_2$. Then we have
$$d_{so}(N,k)\geq d_{so}(N_1,k)= d.$$

 By Proposition \ref{prop-D-double-even}, there are no binary $[N_2,k,d+2]$ SO Griesmer codes. So
$$d_{so}(N,k)\leq d_{so}(N_2,k)\leq d.$$
It turns out that $d_{so}(N,k)=d$ for $N_1\leq N\leq N_2$.
\end{proof}

The following theorem is one of the main results in this paper.
\begin{theorem}\label{thm-1}
Let $m,k,d$ be three nonnegative integers, where $d$ is doubly-even. If $m\geq \min\left\{\left\lceil \frac{(k+2)(k-3)}{2k}\right\rceil,
\left\lceil \sqrt{2k+\frac{1}{4}}-\frac{3}{2}\right\rceil\right\}-1$, then we have
$$d_{so}(N,k)=m2^{k-1}+d,$$
where $m(2^k-1)+g(k,d)\leq N\leq m(2^k-1)+g(k,d+2)$.
\end{theorem}

\begin{proof}
Suppose that $N_1=m(2^k-1)+g(k,d)$ and $D_1= m2^{k-1}+d$. It can be checked that $N_1=g(k,D_1)$.
By Theorem \ref{thm-asymptotic}, there is a binary SO
$[N_1,k,D_1]$ Griesmer code. By Lemma \ref{lem-g},
$$d_{so}(N,k)=m2^{k-1}+d,$$
where $g(k,D_1)\leq N\leq g(k,D_1+2)$. Note that
\begin{align*}
  g(k,D_1+2)= & \sum_{i=0}^{k-1}\left\lceil \frac{D_1+2}{2^i}\right\rceil
  =  \sum_{i=0}^{k-1}\left\lceil \frac{m2^{k-1}+d+2}{2^i}\right\rceil
  =m(2^k-1)+g(k,d+2).
\end{align*}
 This completes the proof.
\end{proof}

\begin{example}
Let $d=4$, $k=7$ and $m\geq \min\left\{\left\lceil \frac{(k+2)(k-3)}{2k}\right\rceil,
\left\lceil \sqrt{2k+\frac{1}{4}}-\frac{3}{2}\right\rceil\right\}-1=2$.
Then $g(7,4)=11$ and $g(7,6)=15$. By Theorem \ref{thm-1}, $d_{so}(N,7)=64m+4$ for $127m+11\leq N\leq 127m+15$. More parameters are given in Table 3.
\end{example}

\begin{center}
\begin{tabular}{ll}
\multicolumn{2}{c}{{\rm Table 3: Binary optimal $[N,7]$ SO codes from Theorem \ref{thm-1}, $m\geq 2$}}\\
\hline
\makebox[0.6\textwidth][l]{$N$}& \makebox[0.15\textwidth][l]{$d_{so}(N,7)$} \\
\hline
$127m, \ldots, 127m+8$ & $64m$\\

$127m+11, \ldots, 127m+15$ &  $64m+4$\\

$127m+18, \ldots, 127m+23$ &  $64m+8$\\

$127m+26, \ldots, 127m+30$ &  $64m+12$\\

$127m+33, \ldots, 127m+39$ &  $64m+16$\\

$127m+42, \ldots, 127m+46$ &  $64m+20$\\

$127m+49, \ldots, 127m+54$ & $64m+24$\\

$127m+57, \ldots, 127m+61$ & $64m+28$\\

$127m+64, \ldots, 127m+71$ &  $64m+32$\\

$127m+74, \ldots, 127m+78$ &  $64m+36$\\

$127m+81, \ldots, 127m+86$ &  $64m+40$\\

$127m+89, \ldots, 127m+93$ &  $64m+44$\\

$127m+96, \ldots, 127m+ 102$ &  $64m+48$\\

$127m+105, \ldots, 127m+109$  & $64m+52$\\

$127m+112, \ldots, 127m+117$ &  $64m+56$\\

$127m+120, \ldots, 127m+124$ & $64m+60$\\
\hline
\end{tabular}
\end{center}

\begin{lem}\label{lem-D-2}
If there is a binary SO $[N,k,D]$ code for $N>2^{k}$ and $D>2$, then there is a binary SO $[N-2,k,D^*\geq D-2]$ code.
\end{lem}

\begin{proof}
Let $C$ be a binary SO $[N,k,D]$ code with generator matrix $G$. Since $N>2^{k}$, the matrix $G$ must have the same two columns. Without loss of generality, we assume that
$$G=[G'~|~v~|~v].$$
It can be checked that $G'G'^T=GG^T=O$. Implying that $G'$ generates a binary SO $[N-2,k,D^*\geq D-2]$ code.
\end{proof}

\begin{lem}\label{lem-g-2}
If there is a binary SO $[g(k,d+4),k,d+4]$ Griesmer code with $g(k,d+4)> 2^k$, then
$$d_{so}(N,k)=d+2,$$
where $g(k,d+2)+1\leq N\leq g(k,d+4)-1$.
\end{lem}

\begin{proof}
Let $N_1=g(k,d+2)+1$ and $N_2= g(k,d+4)$. Then $$N_2-N_1=g(k,d+4)-g(k,d+2)-1=\sum_{i=0}^{k-1}\left\lceil \frac{d+4}{2^i}\right\rceil
-\sum_{i=0}^{k-1}\left\lceil \frac{d+2}{2^i}\right\rceil-1=2.$$
By Lemma \ref{lem-D-2}, there is a binary SO $[N_1,k,d+2]$ code.
Suppose that $N_1\leq N\leq N_2-1$. We have
$$d_{so}(N,k)\geq d_{so}(N_1,k)\geq d+2.$$
On the other hand,
$$d_{so}(N,k)\leq d_{so}(N_2-1,k)\leq 2\left\lfloor \frac{d(N_2-1,k)}{2}\right \rfloor \leq d+2.$$
It turns out that $d_{so}(N,k)=d+2$ for $N_1\leq N\leq N_2-1$.
\end{proof}

The following theorem is one of the main results of the paper.
\begin{theorem}\label{thm-3}
Let $m,k,d$ be three nonnegative integers, where $k\geq 5$ and $d$ is doubly-even. If $m\geq \min\left\{\left\lceil \frac{(k+2)(k-3)}{2k}\right\rceil,
\left\lceil \sqrt{2k+\frac{1}{4}}-\frac{3}{2}\right\rceil\right\}-1$, then we have
$$d_{so}(N,k)=m2^{k-1}+d+2,$$
where $m(2^k-1)+g(k,d+2)+1\leq N\leq m(2^k-1)+g(k,d+4)-1$.
\end{theorem}

\begin{proof}
Let $N_2=m(2^k-1)+g(k,d+4)$ and $D_2=m2^{k-1}+d+4$. It then follows from Remark \ref{remark-1111} that $N_2>2^k$. In addition, it can be checked that
\begin{align*}
  g(k,D_2)= & \sum_{i=0}^{k-1}\left\lceil \frac{D_2}{2^i}\right\rceil
  = \sum_{i=0}^{k-1}\left\lceil \frac{m2^{k-1}+d+4}{2^i}\right\rceil
  =m(2^k-1)+g(k,d+4)=N_2.
\end{align*}
By Theorem \ref{thm-asymptotic}, there is a binary SO $[N_2,k,D_2]$ Griesmer code. By Lemma \ref{lem-g-2}, we have
$$d_{so}(N,k)=m2^{k-1}+d+2,$$
where $g(k,D_2-2)+1\leq N\leq g(k,D_2)-1$.
Note that
\begin{align*}
  g(k,D_2-2)= & \sum_{i=0}^{k-1}\left\lceil \frac{D_2-2}{2^i}\right\rceil
  =  \sum_{i=0}^{k-1}\left\lceil \frac{m2^{k-1}+d+2}{2^i}\right\rceil
  =m(2^k-1)+g(k,d+2).
\end{align*}
This completes the proof.
\end{proof}

\begin{remark}
By Theorems \ref{thm-1} and \ref{thm-3}, we determine the exact value of $d_{so}(n,k)$ where $n$ is large relative to $k$. In other words, we reduce the problem with an infinite number of cases to the problem with a finite number of cases.
\end{remark}

\begin{example}
Let $d=4$, $k=7$ and $m\geq \min\left\{\left\lceil \frac{(k+2)(k-3)}{2k}\right\rceil,
\left\lceil \sqrt{2k+\frac{1}{4}}-\frac{3}{2}\right\rceil\right\}-1=2$.
Then $g(7,6)=15$ and $g(7,8)=18$. By Theorem \ref{thm-1}, $d_{so}(N,7)=64m+6$ for $127m+16\leq N\leq 127m+17$. More parameters are given in Table 4.
\end{example}

\begin{center}
\begin{tabular}{ll}
\multicolumn{2}{c}{{\rm Table 4: Binary optimal $[N,7]$ SO codes from Theorem \ref{thm-3}, $m\geq 2$}}\\
\hline
\makebox[0.5\textwidth][l]{$N$}& \makebox[0.15\textwidth][l]{$d_{so}(N,7)$} \\
\hline
$127m+9,  127m+10$ & $64m+2$\\

$127m+16,  127m+17$ &  $64m+6$\\

$127m+24,  127m+25$ &  $64m+10$\\

$127m+31,  127m+32$ &  $64m+14$\\

$127m+40,  127m+41$ &  $64m+18$\\

$127m+47,  127m+48$ &  $64m+22$\\

$127m+55,  127m+56$ & $64m+26$\\

$127m+62,  127m+63$ & $64m+30$\\

$127m+72, 127m+73$ &  $64m+34$\\

$127m+79,  127m+80$ &  $64m+38$\\

$127m+87,  127m+88$ &  $64m+42$\\

$127m+94,  127m+95$ &  $64m+46$\\

$127m+103,  127m+ 104$ &  $64m+50$\\

$127m+110,  127m+111$  & $64m+54$\\

$127m+118,  127m+119$ &  $64m+58$\\

$127m+125,  127m+126$ & $64m+62$\\
\hline
\end{tabular}
\end{center}

Combining \cite{SO-40}, Tables 3 and 4, we will complete the characterization of binary optimal SO $[n,7]$ codes if we can determine the exact value of $d_{so}(n,7)$ for $41\leq n\leq 253$.

\begin{example}
We construct SO Griesmer codes of dimension $k=7$ by Corollary \ref{cor-SO-SS}.
\begin{itemize}
  \item For $k=7$ and $s=1$.
By Corollary \ref{cor-SO-SS}, there exists a binary SO Griesmer code with parameters
 $$\left[127-\sum_{i=1}^p(2^{u_i}-1),7,64-\sum_{i=1}^p2^{u_i-1}\right],$$
where $7>u_1>u_2>\cdots>u_p\geq 3$ and $\sum_{i=1}^{\min\{2,p\}}u_i\leq 7$.

\begin{center}
Table 5\\
\begin{tabular}{l|l|l}
  \hline
  $p=0$ & $p=1$      &     $p=2$ \\
  \hline
  $[127,7,64]$ &$[120,7,60]$ & $[105,7,52]$ \\
  &$[112,7,56]$ &  \\
  &$[96,7,48]$ &   \\
  &$[64,7,32]$ &\\
  \hline
\end{tabular}
\end{center}
By Corollary \ref{cor-all-so}, the codes in Table 5 are all binary SO $[n,7]$ Griesmer codes with minimum distance at most $64$.
  \item For $k=7$ and $s=2$. By Corollary \ref{cor-SO-SS}, there exists a binary SO Griesmer code with parameters
 $$\left[254-\sum_{i=1}^p(2^{u_i}-1),7,128-\sum_{i=1}^p2^{u_i-1}\right],$$
where $7>u_1>u_2>\cdots>u_p\geq 3$ and $\sum_{i=1}^{\min\{3,p\}}u_i\leq 14$.

\begin{center}Table 6\\
\begin{tabular}{l|l|l|l}
  \hline
  $p=0$ & $p=1$      &     $p=2$ &$p=3$ \\
  \hline
  $[254,7,128]$ &$[247,7,124]$ & $[232,7,116]$ &   $[201,7,100]$\\
                &$[239,7,120]$ &  $[216,7,108]$& $[169,7,84]$ \\
               &$[223,7,112]$ &  $[208,7,104]$ & $[153,7,76]$ \\
              &$[191,7,96]$ &   $[184,7,92]$&    \\
              &&$[176,7,88]$   &   \\
              &&$[160,7,80]$&   \\
  \hline
\end{tabular}
\end{center}
\end{itemize}
\end{example}

\begin{remark}
Combining the above three tables and Theorem \ref{thm-large-griesmer}, we construct a SO Griesmer code with the minimum distance $D$ and dimension $k=7$ from SO Solomon-Stiffler codes, where $D$ is doubly-even, $D\geq 48$ and $D\neq 68,72$.
\end{remark}


\begin{example}
 There are binary SO $[138,7,68]$ and $[145,7,72]$ Griesmer codes (see \cite{k=7,Tor-1983,Tor-Til-1}). Applying Lemma \ref{lem-g} to the two SO Griesmer codes, we know that $d_{so}(N,7)=68$ for $138\leq N\leq 142$ and $d_{so}(N,7)=72$ for $145\leq N\leq 150$.
Applying Lemma \ref{lem-g-2} to the code $C$, we know that $d_{so}(N,7)=66$ for $136\leq N\leq 137$ and $d_{so}(N,7)=70$ for $N\leq 143,144$.
Similarly, applying Lemmas \ref{lem-g} and \ref{lem-g-2} to the binary SO Griesmer codes in Tables 5 and 6, we can determine the exact value of $d_{so}(N,7)$ for $127\leq N\leq 253$.
Therefore, we only consider the exact value of $d_{so}(N,7)$ for $41\leq N\leq 126$.
\end{example}

\begin{remark}
By Corollary \ref{cor-SO-SS}, we can construct a binary SO $[g(8,D),8,D]$ code for $D=128$ or $D\geq 156$ and $D$ is doubly-even. By \cite{Tor-1983} and \cite{Tor-Til-2}, there is a binary SO $[g(8,D),8,D]$ code for $132\leq D\leq 152$ and $D$ is doubly-even.
By Lemmas \ref{lem-g} and \ref{lem-g-2}, we can determine the exact value of $d_{so}(N,8)$ for $N\geq 255$. Combining with \cite{SO-40}, we will complete the characterization of binary optimal SO $[n,8]$ codes if we can determine the exact value of $d_{so}(N,8)$ for $41\leq N\leq 254$.
\end{remark}

\section{The nonexistence of some binary self-orthogonal codes with dimension 7}

In this section, we prove the nonexistence of some binary self-orthogonal codes with dimension 7 by applying the residual code technique.

\begin{prop}\label{prop-nonexistence}
There are no binary $[45,6,22]$, $[53,6,26]$, $[60,6,30]$, $[47,7,22],$ $[71,7,34]$, $[79,7,38],$ $[93,7,46],$ $[102,7,50],$ $[109,7,54], [117,7,58]$, and $[124,7,62]$ SO codes.
\end{prop}

\begin{proof}
Suppose that there is a binary SO $[47,7,22]$ code, then it follows from Lemma \ref{SO-Lemma} that there is a binary linear $[25,6,12]$ code, which contradicts the fact that the largest minimum distance of a binary linear $[25,6]$ code is 11 (see \cite{codetables}).
The proof is similar in other cases, so we omit it. This completes the proof.
\end{proof}

\begin{prop}\label{[125-7]}
There are no binary SO $[125,7,62]$ and $[62,7,30]$ codes.
\end{prop}

\begin{proof}
Let $C$ be a binary SO $[125,7,62]$ code with generator matrix $G$.
The first row ${\bf c}_1$ of $G$ is selected to have minimum weight 62.
Then the code $C_1=Res(C,{\bf c}_1)$ is a binary $[63,6,32]$
Griesmer code, which is the simplex code of dimension $6$. So
$$G\sim\left[
\begin{array}{c|c}
G_2&
\begin{array}{c}
   0\cdots0 \\
  S_6
\end{array}
\end{array}
\right].$$
Since $C$ and $C_1$ are SO, $G_2$ generates a binary SO $[62,7]$ code $C_2$ that contains the all-one vector of length $62$.
Then
$$d(C_2)\geq d(C)-\max\{wt({\bf c})~|~{\bf c}\in C_1\}=62-32=30.$$
By the Griesmer bound, $d(C_2)=30$.
That is to say, if there exists a binary SO $[125,7,62]$ code, then there exists a binary SO $[62,7,30]$ code.

The first row ${\bf c'}_1$ of $G_2$ is selected to have minimum weight 30.
Then it follows from Lemma \ref{SO-Lemma} that the code $C_3=Res(C_2,{\bf c'}_1)$ is a binary $[32,6,16]$
Griesmer code, which is the first Reed-Muller code of dimension $6$. So
$$G_2\sim\left[
\begin{array}{c|c}
G_4&
\begin{array}{c}
   0\cdots0 \\
  R(1,5)
\end{array}
\end{array}
\right]\sim
\left[
\begin{array}{c|c}
G_4&
\begin{array}{c}
   0\cdots0 \\
  S_6\backslash S_5
\end{array}
\end{array}
\right]\sim
\left[
\begin{array}{c|c}
G_4&
  \begin{array}{cc}
  0 & 0\cdots 0\\
    1 & 1\cdots 1 \\
    \begin{array}{c}
      0 \\
      \vdots \\
      0
    \end{array}
     & S_5
  \end{array}
\end{array}
\right].$$
It follows from $d(C_2)\geq 30$ that the first row of $G_4$ is the all-one vector of length $30$.
Since $C_2$ contains the all-one vector of length $62$, the second row of $G_4$ is the all-one vector of length $30$. That is,
$$G_2\sim G_2'=
\left[
\begin{array}{ccc}
  1\cdots 1 & 0 & 0\cdots 0 \\
  1\cdots 1 & 1 & 1\cdots 1 \\
  G_5 & \begin{array}{c}
      0 \\
      \vdots \\
      0
    \end{array} & S_5
\end{array}
\right].$$
By deleting the second row and the $31$-th column of $G_2'$, we obtain the following matrix
$$G_6=\left[
\begin{array}{cc}
  1\cdots 1  & 0\cdots 0 \\
  G_5 & S_5
\end{array}
\right]=\left[
\begin{array}{c|c}
G_7&
\begin{array}{c}
   0\cdots0 \\
  S_5
\end{array}
\end{array}
\right],$$
which generates a binary SO $[61,6,30]$ code $C_6$.
Since $C_6$ is SO and $S_5$ generates a binary SO code, $G_7$ generates a binary SO $[30,6]$ code $C_7$. Then
$$d(C_7)\geq d(C_6)-\max\left\{{\mbox{wt}}({\bf x}\cdot S_5)~|~{\bf x}\in \F_2^5\right\}=30-16=14.$$
This contradicts that $d_{so}(30,6)=12$ (see \cite{SO-40}). Hence there are no binary SO $[125,7,62]$ and $[62,7,30]$ codes.
\end{proof}

\begin{prop}\label{[61-6]}
There is no binary SO $[61,6,30]$ code.
\end{prop}
\begin{proof}
The proof is similar to that of Proposition \ref{[125-7]}. So we omit it.
\end{proof}

\begin{prop}\label{[118-7]}
There is no binary SO $[118,7,58]$ code.
\end{prop}

\begin{proof}
Let $C$ be a binary SO $[118,7,58]$ code with generator matrix $G$.
We first discuss how to select the first three rows in the generator matrix $G$ for the code $C$.
The first row ${\bf c}_1$ is selected to have minimum weight 58. Since $58 \equiv 2 \pmod{4}$ it follows from Lemma \ref{SO-Lemma} that the code $C_0=Res(C,{\bf c}_1)$ has parameters $[60,6,30]$ and is a Griesmer code. It is shown that a code with these parameters has to be a Solomon-Stiffler code with the generator matrix $G_0= [S_6 \setminus S_2]$ (see \cite{Tor-4}).
Without loss of generality, we can assume that
$$G_0= [S_6\backslash S_2]=\left[
\begin{array}{cccc}
\overbrace{1\cdots1}^{15}& \overbrace{1\cdots1}^{15}& \overbrace{0\cdots0}^{15} & \overbrace{0\cdots0}^{15}\\
1\cdots1 &0\cdots0& 1\cdots1 &0\cdots0\\
S_4&S_4&S_4&S_4
\end{array}
\right]=\left[
\begin{array}{c}
  {\bf r}_2 \\
  {\bf r}_3 \\
  {\bf r}_4 \\
  {\bf r}_5 \\
  {\bf r}_6 \\
  {\bf r}_7
\end{array}
\right].$$
Note that the code $C_0$ has weight distribution $\{(0,1),(30,48),(32,15)\}$. It can be checked that
${\bf c}\in C_0$ has weight 32 if and only if ${\bf c}=a_4{\bf r}_4+a_5{\bf r}_5+a_{6}{\bf r}_6+a_7{\bf r}_7$ and there exists $4\leq i\leq 7$ such that $a_i\neq 0$.
Let ${\bf g}_{i+1}$ be the $i$-th row of $G_1$ for $1\leq i\leq 6$. Hence ${\bf c}_{i}=({\bf g}_i,{\bf r}_i)$ is the $i$-th row of $G$ for $2\leq i\leq 7.$
Hence
$$G\sim\left[
\begin{array}{c|c}
  1\cdots 1 & 0\cdots0 \\
  G_1&S_6\backslash S_2
\end{array}
\right]\sim
\left[
\begin{array}{c|cccc}
 \overbrace{1\cdots1}^{58} & \overbrace{0\cdots0}^{15}&\overbrace{0\cdots0}^{15}& \overbrace{0\cdots0}^{15}& \overbrace{0\cdots0}^{15} \\
 \hline
  \multirow{3}{*}{$G_1$} & 1\cdots1 &1\cdots1& 0\cdots0 &0\cdots0\\
       &1\cdots1 &0\cdots0& 1\cdots1 &0\cdots0\\
       &S_4&S_4&S_4&S_4
\end{array}
\right]=\left[
\begin{array}{c|c}
  \overbrace{1\cdots 1}^{58} & \overbrace{0\cdots 0}^{60} \\
  \hline
 {\bf g}_2&{\bf r}_2\\
 {\bf g}_3&{\bf r}_3\\
 {\bf g}_4&{\bf r}_4\\
 {\bf g}_5&{\bf r}_5\\
 {\bf g}_6&{\bf r}_6\\
 {\bf g}_7&{\bf r}_7
\end{array}
\right].$$

Since ${\bf c}_1\cdot {\bf c}_2=0$, ${\bf g}_2$ is even.
Since ${\mbox{wt}}({\bf c}_2)={\mbox{wt}}({\bf g}_2)+{\mbox{wt}}({\bf r}_2)\geq 58$ and ${\mbox{wt}}({\bf c}_1+{\bf c}_2)=(58-{\mbox{wt}}({\bf g}_2))+{\mbox{wt}}({\bf r}_2)\geq 58$, ${\mbox{wt}}({\bf g}_1)=28$ or $30$. We can assume without loss of generality (adding ${\bf c}_1$ to ${\bf c}_2$ if needed) that ${\mbox{wt}}({\bf c}_2)=58$ and
$$G=\left[
\begin{array}{c|c}
  \overbrace{1\cdots 1}^{28}~\overbrace{1\cdots 1}^{30} & \overbrace{0\cdots 0}^{30}~\overbrace{0\cdots 0}^{30} \\
  1\cdots 1~~0\cdots 0&1\cdots 1~~0\cdots 0\\
  \hline
 {\bf g}_3&{\bf r}_3\\
 {\bf g}_4&{\bf r}_4\\
 {\bf g}_5&{\bf r}_5\\
 {\bf g}_6&{\bf r}_6\\
 {\bf g}_7&{\bf r}_7
\end{array}
\right].$$
Similarly, we can assume that ${\mbox{wt}}({\bf g}_3)=28$ (adding ${\bf c}_1$ to ${\bf c}_3$ if needed).
Since
$${\mbox{wt}}({\bf c}_2+{\bf c}_3)={\mbox{wt}}({\bf g}_2+{\bf g}_3)+{\mbox{wt}}({\bf r}_2+{\bf r}_3)={\mbox{wt}}({\bf g}_2+{\bf g}_3)+30\geq 58,$$
${\mbox{wt}}({\bf g}_2+{\bf g}_3)\geq 28$, i.e., $|{\mbox{supp}}({\bf g}_2) \cap {\mbox{supp}}({\bf g}_3)| \leq 14$.
Since
$${\mbox{wt}}({\bf c}_1+{\bf c}_2+{\bf c}_3)=(58-{\mbox{wt}}({\bf g}_2+{\bf g}_3))+{\mbox{wt}}({\bf r}_2+{\bf r}_3)=(58-{\mbox{wt}}({\bf g}_2+{\bf g}_3))+30\geq 58,$$
${\mbox{wt}}({\bf g}_2+{\bf g}_3)\leq 30$, i.e., $|{\mbox{supp}}({\bf g}_2) \cap {\mbox{supp}}({\bf g}_3)| \geq 13$.
Since
$${\bf c}_2\cdot{\bf c}_3={\bf g}_2\cdot {\bf g}_3+{\bf r}_2\cdot {\bf r}_3={\bf g}_2\cdot {\bf g}_3+1=0,$$
$|{\mbox{supp}}({\bf g}_2) \cap {\mbox{supp}}({\bf g}_3)|$ is odd. This implies that
$|{\mbox{supp}}({\bf g}_2) \cap {\mbox{supp}}({\bf g}_3)| =13$. Hence we may assume that
$$G=\left[
\begin{array}{cccc|cccc}
  \overbrace{1\cdots 1}^{13}&\overbrace{1\cdots 1}^{15}&\overbrace{1\cdots 1}^{15}&\overbrace{1\cdots 1}^{15} & \overbrace{0\cdots0}^{15}&\overbrace{0\cdots0}^{15}& \overbrace{0\cdots0}^{15}& \overbrace{0\cdots0}^{15} \\
  1\cdots 1&1\cdots 1&0\cdots 0&0\cdots 0&1\cdots1&1\cdots1&0\cdots0&0\cdots0\\
   1\cdots1&0\cdots0&1\cdots1&0\cdots0&1\cdots1&0\cdots0&1\cdots1&0\cdots0\\
   \hline
 G_{13}&A_1&A_2&A_3&S_4&S_4&S_4&S_4
\end{array}
\right].$$

If we consider the residual $[60,6,30]$ codes $Res(C,{\bf c}_2)$, and $Res(C,{\bf c}_3)$, then we can obtain $A_1=A_2=A_3=[S_4]$. That is to say,
$$G=\left[
\begin{array}{cccc|cccc}
  \overbrace{1\cdots 1}^{13}&\overbrace{1\cdots 1}^{15}&\overbrace{1\cdots 1}^{15}&\overbrace{1\cdots 1}^{15} & \overbrace{0\cdots0}^{15}&\overbrace{0\cdots0}^{15}& \overbrace{0\cdots0}^{15}& \overbrace{0\cdots0}^{15} \\
  1\cdots 1&1\cdots 1&0\cdots 0&0\cdots 0&1\cdots1&1\cdots1&0\cdots0&0\cdots0\\
   1\cdots1&0\cdots0&1\cdots1&0\cdots0&1\cdots1&0\cdots0&1\cdots1&0\cdots0\\
   \hline
 G_{13}&S_4&S_4&S_4&S_4&S_4&S_4&S_4
\end{array}
\right].$$

We next consider the $4 \times 13$ matrix $G_{13}$. Since the code with the generator matrix $G$ generates a binary SO code, the code generated by $G_{13}$ also needs to be SO. We will show that this is impossible and thus we prove the nonexistence of the binary SO $[118,7,58]$ code.
Let ${\bf x}$ be a codeword in the linear code that has generator matrix $G_{13}$.
Let ${\bf y}$ be the corresponding codeword in the linear code that has generator matrix $[S_4,S_4,S_4,S_4,S_4,S_4,S_4]$.
Then it follows that
\begin{align*}
  {\mbox{wt}}({\bf c}_1 + ({\bf x},{\bf y}))
  &= 13 - {\mbox{wt}}({\bf x}) + (56 -3) = 66 - {\mbox{wt}}({\bf x})
  \geq  58~  {\rm and}\\
  {\mbox{wt}}({\bf c}_1 +{\bf c}_2 + ({\bf x},{\bf y}))
  &=  {\mbox{wt}}({\bf x}) + (56 -4) = 52+ {\mbox{wt}}({\bf x})
  \geq  58.
\end{align*}
Hence, in the code $C_{13}$ generated by $G_{13}$  all its nonzero codewords have weight $w$ where $6 \leq w \leq 8$. Since all weights of the codewords in $C$ are even, the same holds for $C_{13}$. Therefore,  $C_{13}$ has only  nonzero codewords having weight $w=6$ and $w=8$. Furthermore, we observe that  $G_{13}$ has rank four since otherwise we can construct a nonzero codeword in $G$ with zeros in its rightmost 13 positions of weight $56$, which is a contradiction. Hence, there exists a binary SO $[13,4,6]$ code, which contradicts that $d_{so}(13,4)=4$ (see \cite{SO-40}).
\end{proof}

\begin{prop}\label{[110,7,54]}
There is no binary SO $[110,7,54]$ code.
\end{prop}

\begin{proof}
Let $C$ be a binary SO $[110,7,54]$ code with generator matrix $G$.
We first discuss how to select the first four rows in the generator matrix $G$ for the code $C$.
The first row ${\bf c}_1$ is selected to have minimum weight 54. Since $54 \equiv 2 \pmod{4}$ it follows from Lemma \ref{SO-Lemma} that the code $C_0=Res(C,{\bf c}_1)$ has parameters $[56,6,28]$ and is a Griesmer code. It is shown that a code with these parameters has to be a Solomon-Stiffler code with the generator matrix $G_0= [S_6 \setminus S_3]$ (see \cite{Tor-4}).
Without loss of generality, we can assume that $G\sim\left[
\begin{array}{c|c}
  1\cdots 1 & 0\cdots0 \\
  G_1&S_6\backslash S_3
\end{array}
\right]$, where
$G_0=[S_6\backslash S_3]=$
$$\left[
\begin{array}{cccccccc}
\overbrace{1\cdots1}^{7}& \overbrace{1\cdots1}^{7}&\overbrace{1\cdots1}^{7}& \overbrace{1\cdots1}^{7}& \overbrace{0\cdots0}^{7} & \overbrace{0\cdots0}^{7}&\overbrace{0\cdots0}^{7} & \overbrace{0\cdots0}^{7}\\
1\cdots1 &1\cdots1 &0\cdots0&0\cdots0& 1\cdots1&1\cdots1  &0\cdots0&0\cdots0\\
1\cdots1 &0\cdots0 &1\cdots1&0\cdots0& 1\cdots1&0\cdots0  &1\cdots1&0\cdots0\\
S_3&S_3&S_3&S_3&S_3&S_3&S_3&S_3
\end{array}
\right]=\left[
\begin{array}{c}
  {\bf r}_2 \\
  {\bf r}_3 \\
  {\bf r}_4 \\
  {\bf r}_5 \\
  {\bf r}_6 \\
  {\bf r}_7
\end{array}
\right].$$
Note that the code $C_0$ has weight distribution $\{(0,1),(28,56),(32,7)\}$. It can be checked that
${\bf c}\in C_0$ has weight 32 if and only if ${\bf c}=a_5{\bf r}_5+a_{6}{\bf r}_6+a_7{\bf r}_7$ and there exists $5\leq i\leq 7$ such that $a_i\neq 0$.
Let ${\bf g}_{i+1}$ be the $i$-th row of $G_1$ for $1\leq i\leq 6$. Hence ${\bf c}_{i}=({\bf g}_i,{\bf r}_i)$ is the $i$-th row of $G$ for $2\leq i\leq 7.$
Hence
\begin{align*}
  G\sim & \left[
\begin{array}{c|cccccccc}
 \overbrace{1\cdots1}^{54} & \overbrace{0\cdots0}^{7}&\overbrace{0\cdots0}^{7}& \overbrace{0\cdots0}^{7}& \overbrace{0\cdots0}^{7}& \overbrace{0\cdots0}^{7}&\overbrace{0\cdots0}^{7}& \overbrace{0\cdots0}^{7}& \overbrace{0\cdots0}^{7} \\
 \hline
  \multirow{4}{*}{$G_1$} & 1\cdots1 &1\cdots1& 1\cdots1 &1\cdots1& 0\cdots0 &0\cdots0& 0\cdots0 &0\cdots0\\
       &1\cdots1&1\cdots1 &0\cdots0&0\cdots0&1\cdots1 &1\cdots1 &0\cdots0&0\cdots0\\
       &1\cdots1 &0\cdots0 &1\cdots1&0\cdots0& 1\cdots1&0\cdots0  &1\cdots1&0\cdots0\\
&S_3&S_3&S_3&S_3&S_3&S_3&S_3&S_3
\end{array}
\right] \\
  = & \left[
\begin{array}{c|c}
  \overbrace{1\cdots 1}^{54} & \overbrace{0\cdots 0}^{56} \\
  \hline
 {\bf g}_2&{\bf r}_2\\
 {\bf g}_3&{\bf r}_3\\
 {\bf g}_4&{\bf r}_4\\
 {\bf g}_5&{\bf r}_5\\
 {\bf g}_6&{\bf r}_6\\
 {\bf g}_7&{\bf r}_7
\end{array}
\right].
\end{align*}

Since ${\bf c}_1\cdot {\bf c}_2=0$, ${\bf g}_2$ is even.
Since
\begin{align*}
  {\mbox{wt}}({\bf c}_2)= &~ {\mbox{wt}}({\bf g}_2)+{\mbox{wt}}({\bf r}_2)={\mbox{wt}}({\bf g}_2)+28\geq 54~{\rm and} \\
  {\mbox{wt}}({\bf c}_1+{\bf c}_2)= &~ (54-{\mbox{wt}}({\bf g}_2))+{\mbox{wt}}({\bf r}_2)=(54-{\mbox{wt}}({\bf g}_2))+28\geq 54,
\end{align*}
${\mbox{wt}}({\bf g}_1)=26$ or $28$. We can assume without loss of generality (adding ${\bf c}_1$ to ${\bf c}_2$ if needed) that ${\mbox{wt}}({\bf c}_2)=54$ and
$$G=\left[
\begin{array}{c|c}
  \overbrace{1\cdots 1}^{26}~\overbrace{1\cdots 1}^{28} & \overbrace{0\cdots 0}^{28}~\overbrace{0\cdots 0}^{28} \\
  1\cdots 1~~0\cdots 0&1\cdots 1~~0\cdots 0\\
  \hline
 {\bf g}_3&{\bf r}_3\\
 {\bf g}_4&{\bf r}_4\\
 {\bf g}_5&{\bf r}_5\\
 {\bf g}_6&{\bf r}_6\\
 {\bf g}_7&{\bf r}_7
\end{array}
\right].$$
Similarly, we can assume that ${\mbox{wt}}({\bf g}_3)=26$ (adding ${\bf c}_1$ to ${\bf c}_3$ if needed).
Since
$${\mbox{wt}}({\bf c}_2+{\bf c}_3)={\mbox{wt}}({\bf g}_2+{\bf g}_3)+{\mbox{wt}}({\bf r}_2+{\bf r}_3)={\mbox{wt}}({\bf g}_2+{\bf g}_3)+28\geq 54,$$
${\mbox{wt}}({\bf g}_2+{\bf g}_3)\geq 26$, i.e., $|{\mbox{supp}}({\bf g}_2) \cap {\mbox{supp}}({\bf g}_3)| \leq 13$.
Since
$${\mbox{wt}}({\bf c}_1+{\bf c}_2+{\bf c}_3)=(54-{\mbox{wt}}({\bf g}_2+{\bf g}_3))+{\mbox{wt}}({\bf r}_2+{\bf r}_3)=(58-{\mbox{wt}}({\bf g}_2+{\bf g}_3))+28\geq 58,$$
we have ${\mbox{wt}}({\bf g}_2+{\bf g}_3)\leq 28$, i.e., $|{\mbox{supp}}({\bf g}_2) \cap {\mbox{supp}}({\bf g}_3)| \geq 12$.
Since
$${\bf c}_2\cdot{\bf c}_3={\bf g}_2\cdot {\bf g}_3+{\bf r}_2\cdot {\bf r}_3={\bf g}_2\cdot {\bf g}_3=0,$$
$|{\mbox{supp}}({\bf g}_2) \cap {\mbox{supp}}({\bf g}_3)|$ is even. This implies that
$|{\mbox{supp}}({\bf g}_2) \cap {\mbox{supp}}({\bf g}_3)| =12$. Hence we may assume that
$$G=\left[
\begin{array}{cccc|cccc}
  \overbrace{1\cdots 1}^{12}&\overbrace{1\cdots 1}^{14}&\overbrace{1\cdots 1}^{14}&\overbrace{1\cdots 1}^{14} & \overbrace{0\cdots0}^{14}&\overbrace{0\cdots0}^{14}& \overbrace{0\cdots0}^{14}& \overbrace{0\cdots0}^{14} \\
  1\cdots 1&1\cdots 1&0\cdots 0&0\cdots 0&1\cdots1&1\cdots1&0\cdots0&0\cdots0\\
   1\cdots1&0\cdots0&1\cdots1&0\cdots0&1\cdots1&0\cdots0&1\cdots1&0\cdots0\\
 \hline
 \multicolumn{4}{c}{{\bf g}_4}&\multicolumn{4}{|c}{{\bf r}_4}\\
  \multicolumn{4}{c}{{\bf g}_5}&\multicolumn{4}{|c}{{\bf r}_5}\\
   \multicolumn{4}{c}{{\bf g}_6}&\multicolumn{4}{|c}{{\bf r}_6}\\
    \multicolumn{4}{c}{{\bf g}_7}&\multicolumn{4}{|c}{{\bf r}_7}
\end{array}
\right].$$

Similarly, we can assume that ${\mbox{wt}}({\bf g}_4)=26$. Then, we can obtain that
$$|{\mbox{supp}}({\bf g}_2) \cap {\mbox{supp}}({\bf g}_4)| =|{\mbox{supp}}({\bf g}_3) \cap {\mbox{supp}}({\bf g}_4)| =12.$$
Suppose that $|{\mbox{supp}}({\bf g}_2) \cap {\mbox{supp}}({\bf g}_3) \cap {\mbox{supp}}({\bf g}_4)| =a$.
So $${\mbox{wt}}({\bf g}_2+{\bf g}_3+{\bf g}_4)=a+(a+2)+(a+2)+(a+2)=4a+6.$$
Since
$${\mbox{wt}}({\bf c}_2+{\bf c}_3+{\bf c}_4)={\mbox{wt}}({\bf g}_2+{\bf g}_3+{\bf g}_4)+{\mbox{wt}}({\bf r}_2+{\bf r}_3+{\bf r}_4)={\mbox{wt}}({\bf g}_2+{\bf g}_3+{\bf g}_4)+28\geq 54,$$
${\mbox{wt}}({\bf g}_2+{\bf g}_3+{\bf g_4})\geq 26$.
Since
\begin{align*}
  {\mbox{wt}}({\bf c}_1+{\bf c}_2+{\bf c}_3+{\bf c}_4)= & (54-{\mbox{wt}}({\bf g}_2+{\bf g}_2+{\bf g}_3))+{\mbox{wt}}({\bf r}_2+{\bf r}_3+{\bf r}_4) \\
  = & (54-{\mbox{wt}}({\bf g}_2+{\bf g}_3+{\bf g}_4))+28\\
  \geq &54,
\end{align*}
${\mbox{wt}}({\bf g}_2+{\bf g}_3+{\bf g}_4)\leq 28$. Then
$26\leq 4a+6\leq 28$, i.e., $a=5$.
Hence we may assume that
$$G=[A~|~B],$$
where
$$A=\left[
\begin{array}{cccccccc}
  \overbrace{1\cdots 1}^{5}&\overbrace{1\cdots 1}^{7}&\overbrace{1\cdots 1}^{7}&\overbrace{1\cdots 1}^{7}&\overbrace{1\cdots 1}^{7}&\overbrace{1\cdots 1}^{7}&\overbrace{1\cdots 1}^{7}&\overbrace{1\cdots 1}^{7}  \\
  1\cdots 1&1\cdots 1&1\cdots 1&1\cdots 1&0\cdots 0&0\cdots 0&0\cdots0&0\cdots0\\
  1\cdots 1&1\cdots1&0\cdots0&0\cdots0&1\cdots1&1\cdots1&0\cdots0&0\cdots0\\
 1\cdots1&0\cdots0&1\cdots1&0\cdots0&1\cdots1&0\cdots0&1\cdots1&0\cdots0\\
  \hline
 G_5&A_1&A_2&A_3&A_4&A_5&A_6&A_7
\end{array}
\right]$$
and
$$B=\left[
\begin{array}{cccccccc}
 \overbrace{0\cdots0}^{7}&\overbrace{0\cdots0}^{7}& \overbrace{0\cdots0}^{7}& \overbrace{0\cdots0}^{7}& \overbrace{0\cdots0}^{7}&\overbrace{0\cdots0}^{7}& \overbrace{0\cdots0}^{7}& \overbrace{0\cdots0}^{7} \\
  1\cdots1 &1\cdots1& 1\cdots1 &1\cdots1& 0\cdots0 &0\cdots0& 0\cdots0 &0\cdots0\\
       1\cdots1&1\cdots1 &0\cdots0&0\cdots0&1\cdots1 &1\cdots1 &0\cdots0&0\cdots0\\
       1\cdots1 &0\cdots0 &1\cdots1&0\cdots0& 1\cdots1&0\cdots0  &1\cdots1&0\cdots0\\
       \hline
S_3&S_3&S_3&S_3&S_3&S_3&S_3&S_3
\end{array}
\right].$$

If we consider the residual $[56,6,28]$ codes $Res(C,{\bf c}_2)$, $Res(C,{\bf c}_3)$ and $Res(C,{\bf c}_4)$, then we can obtain $A_1=A_2=A_3=A_4=A_5=A_6=A_7=[S_3]$. That is to say,
$$A=\left[
\begin{array}{cccccccc}
  \overbrace{1\cdots 1}^{5}&\overbrace{1\cdots 1}^{7}&\overbrace{1\cdots 1}^{7}&\overbrace{1\cdots 1}^{7}&\overbrace{1\cdots 1}^{7}&\overbrace{1\cdots 1}^{7}&\overbrace{1\cdots 1}^{7}&\overbrace{1\cdots 1}^{7}  \\
  1\cdots 1&1\cdots 1&1\cdots 1&1\cdots 1&0\cdots 0&0\cdots 0&0\cdots0&0\cdots0\\
  1\cdots 1&1\cdots1&0\cdots0&0\cdots0&1\cdots1&1\cdots1&0\cdots0&0\cdots0\\
 1\cdots1&0\cdots0&1\cdots1&0\cdots0&1\cdots1&0\cdots0&1\cdots1&0\cdots0\\
  \hline
 G_5&S_3&S_3&S_3&S_3&S_3&S_3&S_3
\end{array}
\right].$$

We next consider the $3 \times 5$ matrix $G_{5}$.
We observe that $G_{5}$ has rank three since otherwise we can construct a nonzero codeword ${\bf c}$ in $G$ with zeros in it leftmost 5 positions of weight $60$. Then
$${\mbox{wt}}({\bf c}+{\bf c_1}+{\bf c_2})=60-7=53<54,$$
which is a contradiction.
Since the generator matrix $G$ generates a SO code, the code generated by $G_{5}$ also needs to be SO. Hence there exists a binary SO $[5,3]$ code. This is impossible and thus we have proved the nonexistence of the binary SO $[110,7,54]$ code.
\end{proof}

\begin{cor}\label{[54-6]}
There is no binary SO $[54,6,24]$ code.
\end{cor}

\begin{proof}
Let $C$ be a binary SO $[54,6,26]$ code with generator matrix $G$.
Similar to the analysis of Proposition \ref{[110,7,54]}, it is not difficult to see that the matrix $G$ is equivalent to the matrix $G_1$ in Proposition \ref{[110,7,54]}. Finally we construct a binary SO $[5,3]$ code if there exists a binary SO $[54,6,26]$ code. This is impossible and thus we have proved the nonexistence of the binary SO $[54,6,26]$ code.
\end{proof}

\begin{prop}\label{prop-46-6-22}
There is no binary SO $[46,6,22]$ code.
\end{prop}

\begin{proof}
Let $C$ be a binary SO $[46,6,22]$ code with generator matrix $G$.
We first discuss how to select the first four rows in the generator matrix $G$ for the code $C$.
Similar to the previous analysis, we may assume that
$$G=[A~|~B],$$
where
$$A=\left[
\begin{array}{cccccccccccccccc}
  1&111&111&111&111&111&111&111  \\
  1&111&111&111&000&000&000&000\\
  1&111&000&000&111&111&000&000\\
  1&000&111&000&111&000&111&000\\
  \hline
 G_2&S_2&S_2&S_2&S_2&S_2&S_2&S_2
\end{array}
\right]$$
and
$$B=\left[
\begin{array}{cccccccccccccccc}
  000&000&000&000&000&000&000&000  \\
  111&111&111&111&000&000&000&000\\
  111&111&000&000&111&111&000&000\\
  111&000&111&000&111&000&111&000\\
  \hline
 G_5&S_2&S_2&S_2&S_2&S_2&S_2&S_2
\end{array}
\right].$$

It can be checked that the last five rows of the matrix $A$ generates a binary SO $[22,5,10]$ code. This contradicts that $d_{so}(22,5)=8$ (see \cite{SO-40}). Thus we prove the nonexistence of the binary SO $[46,6,22]$ code.
\end{proof}

\begin{remark}
According to \cite{Kim-SO} and \cite{SO-5-6}, there are optimality of six binary SO $[n,6]$ codes which were not discussed, namely, $n=45,46,53,54,60,61$.
Combining Proposition \ref{prop-nonexistence}, Proposition \ref{[61-6]}, Corollary \ref{[54-6]} and Proposition \ref{prop-46-6-22}, we completely solve the remaining case in \cite{Kim-SO} and \cite{SO-5-6}.
\end{remark}

\begin{prop}\label{[94-7]}
There is no binary SO $[94,7,46]$ code.
\end{prop}

\begin{proof}
Let $C$ be a binary SO $[94,7,46]$ code with generator matrix $G$.
We first discuss how to select the first five rows in the generator matrix $G$ for the code $C$.
Similar to the previous analysis, we may assume that
$$G=[A~|~B],$$
where
{\small$$A=\left[
\begin{array}{cccccccccccccccc}
  1&111&111&111&111&111&111&111&111&111&111&111&111&111&111&111  \\
  1&111&111&111&111&111&111&111&000&000&000&000&000&000&000&000\\
  1&111&111&111&000&000&000&000&111&111&111&111&000&000&000&000\\
  1&111&000&000&111&111&000&000&111&111&000&000&111&111&000&000\\
  1&000&111&000&111&000&111&000&111&000&111&000&111&000&111&000\\
  \hline
 G_2&S_2&S_2&S_2&S_2&S_2&S_2&S_2&S_2&S_2&S_2&S_2&S_2&S_2&S_2&S_2
\end{array}
\right]$$}
and
{\small$$B=\left[
\begin{array}{cccccccccccccccc}
  000&000&000&000&000&000&000&000&000&000&000&000&000&000&000&000  \\
  111&111&111&111&111&111&111&111&000&000&000&000&000&000&000&000\\
  111&111&111&111&000&000&000&000&111&111&111&111&000&000&000&000\\
  111&111&000&000&111&111&000&000&111&111&000&000&111&111&000&000\\
  111&000&111&000&111&000&111&000&111&000&111&000&111&000&111&000\\
  \hline
 G_2&S_2&S_2&S_2&S_2&S_2&S_2&S_2&S_2&S_2&S_2&S_2&S_2&S_2&S_2&S_2
\end{array}
\right].$$}

It can be checked that the last six rows of the matrix $A$ generates a binary SO $[46,6,22]$ code. This contradicts that $d_{so}(46,6)=20$ (see Proposition \ref{prop-46-6-22}). Thus we prove the nonexistence of the $[94,7,46]$ SO code.
\end{proof}

\begin{prop}\label{[72-7]}
There is no binary SO $[72,7,34]$ code.
\end{prop}

\begin{proof}
The first part of the proof will reconstruct the $65$ columns of the generator matrix of the $[72,7,34]$ code if it exists. The second part of the proof will show that it is impossible to complete the construction to all 72 rows of the generator matrix.

Let $C$ be a binary SO $[72,7,34]$ code.
Then it follows from Lemma \ref{SO-Lemma} that the residual code $C_0 = Res(C,{\bf c}_1)$ with respect to the codeword ${\bf c}_1 \in C$ of minimum weight 34 has parameters $[38,6,18]$. 
This must be either a Solomon-Stiffler code or a Belov code.

Note that for the residual code $d = 2^5 - 2^3 - 2^2 - 2=18$, i.e., $k = 6, u_1 = 4, u_2 = 3, u_3 = 2$ implies that there is no code of the Solomon and Stiffler type since $u_1 + u_2 > k$. However, there is a binary $[38,6,18]$ code of Belov type of the form $G_0 = [S_6 \setminus S_5 | U ]$ where $U \subset S_5$ generates a $[6,5,2]$ code. Note the first 32 positions form a first-order Reed-Muller code.
Note that the $[38,6,18]$ code of Belov type is defined by the following generator matrix
\[
 G_0 = \left[\begin{array}{ccccccr}
 \overbrace{\hat{S_5}}^{32}                        &            \overbrace{ T_5}^6                         \\
 1 1 \cdots 1                     &         0 0 0 0 0 0                  \\
                          \end{array} \right].
 \]
where the first 32 positions generates a first-order $[32,6,16]$ Reed-Muller code and $T_5$ generates a $[6,5,2]$ code. Furthermore, $\hat{S}_5$ denotes the set of all 5-dimensional  column vectors. Thus if we identify the positions of the codeword of weight 32 among the 38 positions, the code is easily constructed since the remaining part is a  $[6,5,2]$ code with generator matrix $T_5$ given by,

\[
 T_5 = \left[\begin{array}{ccccccr}
                             1  &  1  &  0  & 0 & 0 & 0   \\
                             1  &  0  &  1  & 0 & 0 & 0   \\
                             1  &  0  &  0 &  1 & 0 & 0   \\
                             1  &  0  &  0 &  0 & 1 & 0   \\
                             1  &  0  &  0 &  0 & 0 & 1   \\
 \end{array} \right].
 \]

Here, the 38 rightmost positions generate the unique $[38,6,18]$ Griesmer code up to equivalence. The code has a codeword of weight 32 that we assume  in the following are in the last row of the generator matrix.
The columns in these 32 positions are distinct and thus the rightmost 38 columns of $C$ are known due to the parameters of $Res(C,{\bf c}_1)$,
\[
 G = \left[\begin{array}{ccccccr}
  \overbrace{1 1 . . . 1 }^{34}  &  \overbrace{0 0 . . . 0}^{32} &\overbrace{ 0 0 . . . 0}^{6}                  \\
  \multirow{2}{*}{$U$}               &           \hat{S_5}     &                T_5        \\
             &             1 1 ... 1        &           0 0 ... 0
 \end{array} \right].
 \]

To construct explicitly the next 18 columns (columns number 17-34 from the left) we are considering the first two rows of the generator matrix of $C$. Then since both $Res(C,{\bf c}_1)$ and $Res(C,{\bf c}_2)$ have the same parameters as a $[38,6,18]$ code we first need to identify the codeword of weight 32 in $Res(C,{\bf c}_2)$. The only possible codeword of weight 32 is the last row since all other codewords have at least 8 zeros when restricted to the rightmost 38 positions. Hence, we can assume without loss of generality that,

\[
 G = \left[\begin{array}{ccccc}

 \overbrace{1 1 . . . 1 }^{16} &  \overbrace{1 1 . . . 1\;1 1}^{18} & \overbrace{0 0 . . . 0  \: 0 0}^{18}   &             \overbrace{0 0 . . . 0}^{16} &  \overbrace{0 0 0 0}^4                                \\
 1 1 . . . 1                   &             0 0 . . . 0    \;  0 0              &              1 1 . . . 1      \; 1 1   &             0 0 . . . 0                  &  0 0 0 0                                 \\
\multirow{2}{*}{$U_5$}    & \hat{S_4}  \; \;  \; \;  \;  \; .. &    \hat{S_4} \; \;  \; \;  \;  \; ..  &             \hat{S_4 }                &  I_4                                        \\
 &            1 1 ...  1  \; 0 0            &            1 1 \ ...  1  \;  0  0  &  1 1 ... 1                       &  0 0 0 0
 \end{array} \right].
 \]
Here, $I_4$ is the $4 \times 4$ identity matrix. Adding the first (or second) row if necessary we assume that column 34 is the unit vector $(1 0 \cdots 0)^{T}$ and column 52 is the unit vector $(0 1 0 \cdots 0)^{T}$. Note that $Res(C,{\bf c}_1)$ and $Res(C,{\bf c}_2)$ are $[38,6,18]$ codes meeting the Griesmer bound with even minimum distance. Therefore, all codewords have even weights and it follows that  column 33 is $(1 0 1 1 1 1 0)^{T}$ and column 51 is the vector $(0 1 1 1 1 1 1 0)^{T}$. Note that the 4 rightmost columns together with column 51, 52 form a $[6,5,2]$ even weight code. The same is the case for column 33, 34 and the four rightmost columns.

Hence we have therefore uniquely constructed 56 columns of (an equivalent) generator matrix for the $[38,6,18]$ code. Continuing this argument by considering the first three rows by selecting a codeword of weight 9 in the $[20,5,9]$  residual code of a codeword of weight 18 in $Res(C,{\bf c}_1)$, we can use the fact the $Res(C,{\bf c}_3)$ also is a $[38,6,18]$ code. In this we can assume without loss of generality,

\[
 G = \left[\begin{array}{ccccccccccccccr}

 \overbrace{1 1 . . . 1 }^7                &            \overbrace{ 1 1 . . . 1     \; 1}^9      &           \overbrace{1 1 . . . 1   \;1 }^9      &   \overbrace{1 1 . . . 1   \; 1}^9
 &  \overbrace{0 0 . . . 0  \; 0 }^9         &    \overbrace{ 0 0 . . . 0     \; 0}^9      &            \overbrace{0 0 . . . 0   \;0}^9      &  \overbrace{ 0 0 . . . 0}^8&   \overbrace{0 0 0}^3                              \\
 1 1 . . . 1                 &             1 1 . . . 1     \; 1      &            0 0 . . . 0    \;0    &             0 0 . . . 0  \; 0
 &             1 1 . . . 1   \;  1        &            1 1 . . . 1     \; 1       &            0 0 . . . 0    \; 0    &            0 0 . . . 0                  &   0 0 0                               \\
 1 1 . . . 1                 &             0 0 . . . 0    \; 0       &            1 1 . . . 1    \; 1    &             0 0 . . . 0 \; 0
 &             1 1 . . . 1   \; 1         &             0 0 . . . 0   \;  0       &            1 1 . . . 1   \; 1     &             0 0 . . . 0                 &   0 0 0                               \\
 & \hspace{1.15cm} 1  &  \hspace{1.15cm} 1   &   \hspace{1.15cm} 0   &  \hspace{1.15cm} 1   & \hspace{1.15cm} 0   &  \hspace{1.15cm} 0    &  &  1 0 0 \\
 \multirow{2}{*}{$U_4$}                    &      \hat{S_3}   \hspace{0.75cm} 1  & \hat{S_3}    \hspace{0.75cm} 1        &             \hat{S_3}   \hspace{0.75cm} 0
 &    \hat{S_3}   \hspace{0.75cm} 1           &            \hat{S_3}  \hspace{0.75cm} 0              &             \hat{S_3}   \hspace{0.75cm} 0        &    \hat{S_3}    &   0 1 0            \\
 & \hspace{1.15cm} 1  &  \hspace{1.15cm} 1   &   \hspace{1.15cm} 0   &  \hspace{1.15cm} 1   & \hspace{1.15cm} 0   &  \hspace{1.15cm} 0    &  &  0 0 1 \\
 &         1 1 ...  1  \; 0            &            1 1 .... 1 \;0      &          1 1 ... 1   \; 0 &                1 1 ... 1 \;0          &         1 1 ...  1 \;  0            &             1 1 ... 1 \;0      &          1 1 ...  1              &     0 0 0
 \end{array} \right].
 \]

Let
\begin{center}
$D = \left[\begin{array}{ccccccccccr}
                               \hat{S_6}      \\
                               1 \cdots 1 \\
       \end{array} \right]
$,
$B = \left[\begin{array}{ccccccccccr}
                             1 1 1 1 1 1 1 1 \\
                             1 1 1 1 1 1 1 1 \\
                             1 1 1 1 1 1 1  1 \\
                              \hat{S_3}          \\
                             1 1 1 1 1 1 1 1
       \end{array} \right]$
and
$A = \left[\begin{array}{ccccccccccr}
                             0    0  0  0  0  0  1  1  1  \\
                             1    1  0  0  0  0  1  0  0  \\
                             1    0  1  0  0  0  0  1  0  \\
                             1    0  0  1  0  0  1  1  0  \\
                             1    0  0  0  1  0  1  1  0  \\
                             1    0  0  0  0  1  1  1  0  \\
                             0    0  0  0  0  0  0  0  0
         \end{array} \right].$
\end{center}

Then
$$G = [ D  \setminus B~|~A~|~U ],$$
where $U$ are the 7 leftmost columns in $G$ and we have proved the following lemma.

The second part of the proof will show that a $[72,7,34]$ SO code do not exist by finding conditions on $U$ and showing that it is impossible to complete the construction of the generator matrix of $G$.

Note that the code with generator matrix is $[D \setminus B]$ is a SO code with possible nonzero weights 24, 28 or 32. It is therefore sufficient to prove that $G$ cannot generate a SO $[72,7,34]$ code. We will find some restrictions of the weights of some linear combinations of the rows in $U$.
Let ${\bf c}_i$ denote the $i$-th row of the generator matrix $G$. Let ${\bf d}_i$, ${\bf a}_i$ and ${\bf u}_i $ denote its restrictions to $D \setminus B$, $A$ and $U$ respectively.
Let ${\bf c}$ be a codeword of $C$. Then
$${\mbox{wt}}({\bf c}) = {\mbox{wt}}({\bf d}) + {\mbox{wt}}({\bf a}) + {\mbox{wt}}({\bf u}).$$

Consider, in particular, the rows ${\bf u}_4, {\bf u}_5$ and ${\bf u}_6$. We will find conditions on the weight of these rows and show that it is impossible to construct a SO code,  this leads to a contradiction.
Since $C$ is a SO code,
\begin{align}\label{eq-5}
 \langle {\bf c}_i,{\bf c}_j\rangle&=\langle {\bf d}_i,{\bf d}_j\rangle+\langle {\bf a}_i,{\bf a}_j\rangle+\langle {\bf u}_i,{\bf u}_j\rangle=\langle {\bf d}_i,{\bf d}_j\rangle+\langle {\bf a}_i,{\bf a}_j\rangle.
\end{align}
Hence, $G$ generates a SO code if and only if the matrix $[A~|~U]$ generates a SO code.
Next we study the $7 \times 7$ matrix $U$ corresponding to the leftmost 7 rows in $G$.

(i) The rows ${\bf u}_i$ for $i=4,5,6$ have ${\mbox{wt}}({\bf u}_i) \geq 2$. This follows from the relation ${\mbox{wt}}({\bf c}_i) = {\mbox{wt}}({\bf d}_i) + {\mbox{wt}}({\bf a}_i) + {\mbox{wt}}({\bf u}_i) \geq 34$  which since ${\mbox{wt}}({\bf d}_i) = 32 - 4 = 28$ and ${\mbox{wt}}({\bf a}_i) = 4$ implies that ${\mbox{wt}}({\bf u}_i) \geq 2$.

(ii) For ${\bf u}_i+ {\bf u}_j$ for distinct $i$ and $j$ in $\{4,5,6\}$, we have ${\mbox{wt}}({\bf u}_i+{\bf u}_j) \geq 4$. This follows similarly from ${\mbox{wt}}({\bf c}_i + {\bf c}_j) = {\mbox{wt}}({\bf d}_i+{\bf d}_j) + {\mbox{wt}}({\bf a}_i + {\bf a}_j)  + {\mbox{wt}}({\bf u}_i + {\bf u}_j) \geq 34$ which since
    ${\mbox{wt}}({\bf d}_i+{\bf d}_j)= 28$ and ${\mbox{wt}}({\bf a}_i + {\bf a}_j) = 2$ implies that ${\mbox{wt}}({\bf u}_i+{\bf u}_j) \geq 4$.

(iii) We have  ${\mbox{wt}}({\bf u}_i) \in  \{2,4\}$ for $i = 4, 5, 6$. This follows since a word ${\mbox{wt}}({\bf u}_i)$ of weight 6 is impossible, since
\begin{eqnarray*}
                 {\mbox{wt}}({\bf c}_3 + {\bf c}_i)  &    =    & {\mbox{wt}}({\bf d}_3 +{\bf  d}_i) + {\mbox{wt}}({\bf a}_3 +{\bf a}_i)  + {\mbox{wt}}({\bf u}_3+{\bf u}_i)    \\
                                      &    =    &    28 + 3  + 7 - {\mbox{wt}}({\bf u}_i)  \\
                                      &     =   &  38 - {\mbox{wt}}({\bf u}_i)
\end{eqnarray*}
which implies ${\mbox{wt}}({\bf u}_i) \leq 4$ since $C$ has minimum distance 34. Since, all codewords in ${\bf c}_i$, ${\bf d}_i$ and ${\bf a}_i$ have even weights for $i=4,5,6$, it follows that all vectors ${\bf u}_i$ have even weights for $i = 4, 5, 6$.

(iv) We have ${\mbox{wt}}({\bf u}_7) = 6$ since,
\begin{eqnarray*}
                 {\mbox{wt}}({\bf c}_1 +  {\bf c}_2 + {\bf  c}_7)  &    =    & {\mbox{wt}}({\bf d}_1 +{\bf d}_2 + {\bf d}_7) + {\mbox{wt}}({\bf a}_1 +{\bf  a}_2 +{\bf a}_7)  + {\mbox{wt}}({\bf u}_1 +{\bf u}_2 + {\bf u}_7)    \\
                                                    &     =   &  28 + {\mbox{wt}}({\bf u}_7)
\end{eqnarray*}
which implies ${\mbox{wt}}({\bf u}_7) = 6$ since $C$ has the minimum distance 34 and all codewords in the SO code $C$ have even weights.

{\bf Case 1:}
Assume ${\mbox{wt}}({\bf u}_i) = {\mbox{wt}}({\bf u}_j) = 2$ for distinct $i$ and $j$ in $\{4,5,6\}$. In this case it follows from (i), (ii) and (iii) above that ${\mbox{wt}}({\bf u}_i + {\bf u}_j)=4$. Hence $|{\mbox{supp}}({\bf u}_i)\cap {\mbox{supp}}({\bf u}_j)|=0.$
This implies that the inner product between rows $i$ and $j$ in $[A~|~ U]$ is the same as inner product between ${\bf a}_i$ and ${\bf a}_j$ which is seen to be odd.
This means the inner product between ${\bf c}_i$ and ${\bf c}_j$ equals the inner product between ${\bf a}_i$ and ${\bf a}_j$ which since being odd contradicts that $C$ is SO.

{ \bf Case 2:}
Assume all three ${\bf u}_i$, $i=4,5,6$ have weight 4. Note that $\langle {\bf u}_i,{\bf u}_j\rangle=1$ for distinct $i$ and $j$ in $\{4,5,6\}$.
Then it follows from (i), (ii), (iii) and the equation (\ref{eq-5}) that
\begin{center}
$|{\mbox{supp}}({\bf u}_i)\cap {\mbox{supp}}({\bf u}_j)| =1$ for distinct $i$ and $j$ in $\{4,5,6\}$.
\end{center}
Hence
\begin{align*}
  |{\mbox{supp}}({\bf u}_4)\cup {\mbox{supp}}({\bf u}_5)\cup {\mbox{supp}}({\bf u}_6)| = &\sum_{i=4}^6 |{\mbox{supp}}({\bf u}_i)|- \sum_{i,j\in\{4,5,6\},i\neq j}|{\mbox{supp}}({\bf u}_i)\cap {\mbox{supp}}({\bf u}_j)|  \\
   & +|{\mbox{supp}}({\bf u}_4)\cap {\mbox{supp}}({\bf u}_5)\cap {\mbox{supp}}({\bf u}_6)|\\
  = &9+|{\mbox{supp}}({\bf u}_4)\cap {\mbox{supp}}({\bf u}_5)\cap {\mbox{supp}}({\bf u}_6)|.
\end{align*}
This contradicts that $|{\mbox{supp}}({\bf u}_4)\cup {\mbox{supp}}({\bf u}_5)\cup {\mbox{supp}}({\bf u}_6)|\leq 7$.

{ \bf Case 3:}
Let ${\mbox{wt}}({\bf u}_{i_1})=2$ and ${\mbox{wt}}({\bf u}_{i_2})={\mbox{wt}}({\bf u}_{i_3})=4$, where $\{i_1,i_2,i_3\}=\{4,5,6\}$. Note that $\langle {\bf u}_i,{\bf u}_j\rangle=1$ for distinct $i$ and $j$ in $\{4,5,6\}$. Then it follows from (i), (ii), (iii) and the equation (\ref{eq-5}) that
\begin{center}
$|{\mbox{supp}}({\bf u}_i)\cap {\mbox{supp}}({\bf u}_j)| =1$ for distinct $i$ and $j$ in $\{4,5,6\}$.
\end{center}
 However, ${\bf u}_7$ has weight 6 and must in this case have an odd inner product with ${\bf u}_4, {\bf u}_5$ and ${\bf u}_6$. Therefore we conclude that the code $C$ does not exist.
\end{proof}

\begin{remark}
Although we have proved the nonexistence of some binary SO $[n,7]$ codes, it is easy to check these results hold for general $k$. The proof process is similar, so it will not be repeated here.
\end{remark}

\begin{example}\label{example-SO}
We start from a binary SO $[18,6,8]$ code (see \cite{SO-40}). By applying (2) of Remark \ref{rem-short}
we can construct a binary SO [82,7,40] code with generator matrix
$$\left[
\begin{array}{c|c|c}
  1 & 1\cdots 1 &000000000000000000 \\
  \hline
  0 &  & 1 0 0 0 0 0 1 0 1 0 1 0 1 1 1 0 0 1\\
  0 &  & 0 1 0 0 0 0 1 1 1 1 0 0 0 1 0 0 1 1 \\
  0 & S_6 & 0 0 1 0 0 0 1 1 0 1 1 1 0 0 0 1 1 0\\
  0 &  & 0 0 0 1 0 0 0 1 1 0 1 1 1 0 0 0 1 1 \\
  0 &  & 0 0 0 0 1 0 1 0 0 1 0 0 1 1 1 1 1 0 \\
  0 &  & 0 0 0 0 0 1 0 1 0 0 1 0 0 1 1 1 1 1
\end{array}\right].
$$

Similarly, we start from binary SO $[16,6,6]$, $[18,6,8]$, $[24,6,10]$, $[26,6,12]$, $[40,6,18]$, $[47,6,22]$, $[55,6,26]$ and $[62,6,30]$ codes. By applying (2) of Remark \ref{rem-short}
we can construct binary SO $[80,7,38]$, [82,7,40], [88,7,42], [90,7,44], $[104,7,50]$, $[111,7,54]$, [119,7,58] and [126,7,62] codes.

In addition, by best-known linear codes (BKLC) database of Magma, there are binary SO $[43,7,20]$, $[44,7,20]$, $[50,7,24]$, $[59,7,28]$, $[75,7,36]$ and $[82,7,40]$ codes.
By puncturing the binary SO $[44,7,20]$ code on $\{9,44\}$, we can construct a binary SO $[42,7,18]$ code. By puncturing the binary SO $[50,7,24]$ code on $\{17,34\}$, we can construct a binary SO $[48,7,22]$ code. By puncturing the binary SO $[59,7,28]$ code on $\{5,6\}$, we can construct a binary SO $[57,7,26]$ code. By puncturing the binary SO $[75,7,36]$ code on $\{12,67\}$, we can construct a binary SO $[73,7,34]$ code.

Through a computer search, we construct a binary $[95,7,46]$ SO code with the generator matrix $G=[I_7~|~M]$, where $M=$
{\footnotesize
$$\left[
\begin{array}{cccc}
  1 0 0 1 0 1 1 0 0 0 1 0 1 1 1 1 0 0 1 0 0 1 0 0 0 1 1 1 1 0 1 1 0 1 0 0 1 0 1 1 1 0 0 1 1 1 0 0 0 0 0 1 0 1 0 0 1 1 0 1 0 1 1 1 0 1 1 0 1 0 0 0 1 0 1 0 1 1 1 1 0 1 0 0 1 0 1 0\\
  1 1 1 1 0 1 0 0 1 0 0 0 0 0 0 1 1 0 0 0 0 0 1 1 1 0 0 1 1 0 1 1 1 1 0 1 1 0 1 0 1 1 0 1 0 1 0 1 0 1 0 1 1 0 1 1 0 1 1 0 1 1 0 0 1 0 1 1 0 1 1 1 0 0 1 1 0 0 1 1 0 1 0 0 1 1 1 1\\
  0 0 0 0 1 1 0 0 0 1 0 1 0 0 1 1 0 0 1 1 0 0 0 0 1 0 0 0 0 1 1 1 0 1 1 1 1 1 1 0 1 0 1 1 1 0 1 1 0 0 1 1 1 1 1 1 0 1 0 0 0 0 1 0 1 0 1 0 1 0 0 0 0 1 0 1 1 0 1 1 1 1 1 1 1 0 1 1\\
  0 1 0 0 1 1 1 1 1 1 1 0 0 0 1 0 1 1 1 0 1 0 1 0 1 0 0 1 1 1 0 1 0 0 0 0 0 1 1 1 1 1 1 1 0 0 1 0 0 1 0 0 0 0 1 1 1 0 1 1 1 1 0 0 0 1 0 0 1 0 0 0 1 0 1 1 0 0 1 0 0 1 1 1 1 1 0 1\\
  1 0 0 0 0 0 0 1 0 1 1 0 1 0 1 0 0 0 0 1 1 0 0 0 1 1 1 1 0 1 1 1 1 0 0 0 0 1 1 1 0 1 0 1 1 1 1 1 1 1 1 1 1 1 1 0 1 1 1 1 1 1 0 0 0 0 0 1 1 0 1 1 0 1 0 0 0 1 1 1 1 0 0 0 0 0 0 0\\
  1 1 1 0 0 0 0 0 0 0 1 1 1 1 0 0 0 1 0 0 1 1 1 0 1 0 0 0 0 1 1 1 1 1 1 0 0 0 0 0 1 0 0 1 0 0 1 0 1 0 0 1 1 1 0 0 0 0 1 1 1 1 0 0 1 0 1 0 1 1 0 1 1 1 1 1 1 1 0 0 1 1 1 1 0 0 1 0\\
  0 1 1 0 0 1 1 1 0 1 1 0 0 1 1 0 1 1 0 1 0 0 0 1 1 0 1 0 1 0 1 1 0 0 1 0 1 0 1 1 1 0 1 0 0 1 1 1 1 1 0 0 0 1 0 1 0 0 1 0 0 1 1 1 1 0 1 0 0 0 1 1 1 0 0 0 0 1 0 1 1 0 0 1 0 1 0 0
\end{array}\right].
$$}
\end{example}

\begin{remark}
We refer to an important upper bound of $d_{so}(n,k)$, namely, $d_{so}(n,k)\leq \left\lfloor\frac{d(n,k)}{2}\right\rfloor$. The upper bound of $d(n,k)$ can be seen in \cite{codetables}.
Applying Lemma \ref{lem-g} to the binary SO Griesmer codes in Table 6, we can partially determine the exact value of $d_{so}(n,7)$ for $41\leq n\leq 126$. Combining Propositions \ref{prop-nonexistence}, \ref{[125-7]}, \ref{[118-7]}, \ref{[110,7,54]}, \ref{[94-7]}, \ref{[72-7]} and Example \ref{example-SO}, we give Table 7.
\end{remark}

\begin{center}
\begin{tabular}{ll|ll}
\multicolumn{4}{c}{{\rm Table 7: Some binary optimal $[N,7]$ SO codes}}\\
\hline
\makebox[0.1\textwidth][l]{$N$}& \makebox[0.15\textwidth][l]{$d_{so}(N,7)$}&
\makebox[0.1\textwidth][l]{$N$}& \makebox[0.1\textwidth][l]{$d_{so}(N,7)$} \\
\hline
$41$ & 16-18&$82,83,84,85,86$ & 40 \\
$42$ & 18&$87$ & 40-42\\

$43,44,45,46,47$ & 20 &$88,89$ & 42\\

$48,49$ & 22&$90,91,92,93,94^*$ & 44\\

$50,51,52,53,54,55$ & 24 &$95$ & 46 \\

$56$ & 24-26&$96,97,98,99,100,101,102$ & 48\\

$57,58$ & 26&$103$ & 48-50\\
$59,60,61,62^*$ & 28 &$104$ & 50\\

$63$ & 28-30&$105,106,107,108,109,110^*$ & 52\\

$64,65,66,67,68,69,70,71$ & 32&$111$ & 54\\

$72^*$ & 32&$112,113,114,115,116,117,118^*$ & 56 \\

$73,74$ & 34&$119$ & 58\\

$75,76,77,78,79$ & 36 &$120,121,122,123,124, 125^*$ & 60\\

$80,81$ & 38&$126$ & 62\\
\hline
\end{tabular}
\end{center}

\section{The nonexistence of some binary self-orthogonal codes with dimension 8}

In this section, we prove the nonexistence of some binary self-orthogonal codes with dimension 8 by applying the residual code technique.

\begin{prop}\label{prop-nonexistence}
There are no binary $[48,8,22]$, $[105,8,50]$, $[112,8,54]$, $[136,8,66]$, $[144,8,70]$, $[152,8,74]$, $[159,8,78]$, $[168,8,82]$, $[175,8,86]$, $[189,8,94]$, $[199,8,98]$, $[214,8,106]$, $[221,8,110]$, $[230,8,114]$, $[237,8,118]$, $[245,8,122]$ and $[252,8,126]$ SO codes.
\end{prop}

\begin{proof}
Suppose that there is a binary SO $[252,8,126]$ code, then it follows from Lemma \ref{SO-Lemma} that there is a binary linear $[126,7,64]$ code, which contradicts the fact that the largest minimum distance of a binary linear $[126,7]$ code is 11 (see \cite{codetables}).
The proof is similar in other cases, so we omit it. This completes the proof.
\end{proof}

\begin{prop}\label{[253-8]}
There are no binary SO $[253,8,126]$ and $[126,8,62]$ codes.
\end{prop}

\begin{proof}
The proof is similar to the proof of Proposition \ref{[125-7]}, so we omit it.
\end{proof}

\begin{prop}\label{[246-8]}
There are no binary SO $[246,8,122]$ codes.
\end{prop}

\begin{proof}
The proof is similar to the proof of Proposition \ref{[118-7]}, so we omit it.
\end{proof}

\begin{prop}\label{[238-8]}
There are no binary SO $[190,8,94]$, $[222,8,110]$, $[238,8,118]$ codes.
\end{prop}

\begin{proof}
The proof is similar to the proof of Proposition \ref{[110,7,54]}, so we omit it.
\end{proof}

\begin{example}\label{example-SO-2}
We start from a binary SO $[18,6,8]$ code (see \cite{SO-40}). By applying (2) of Remark \ref{rem-short}
we can construct a binary SO [254,8,126] code.

Similarly, we start from binary SO $[14,7,4]$, $[18,7,6]$, $[19,7,8]$, $[26,7,10]$, $[27,7,12]$, $[42,7,18]$, $[43,7,20]$, $[48,7,22]$, $[50,7,24]$, $[57,7,26]$, $[59,7,28]$, $[73,7,34]$, $[75,7,36]$, $[80,7,38]$, $[82,7,40]$, $[88,7,42]$, $[95,7,46]$, $[104,7,50]$, $[111,7,54]$, $[119,7,58]$ and $[126,7,62]$ codes. By applying (2) of Remark \ref{rem-short}
we can construct binary SO $[142,8,68]$, $[146,8,70]$, $[147,8,72]$, $[154,8,74]$, $[155,8,76]$, $[170,8,82]$, $[171,8,84]$, $[180,8,86]$, $[178,8,88]$, $[185,8,90]$, $[187,8,92]$, $[201,8,98]$, $[203,8,100]$, $[208,8,102]$, $[210,8,104]$, $[216,8,106]$, $[223,8,110]$, $[232,8,114]$, $[239,8,118]$, [247,8,122] and [254,8,126] codes.

In addition, by best-known linear codes (BKLC) database of Magma, there are binary SO $[45,8,20]$, $[51,8,24]$, $[61,8,28]$, $[62,8,28]$, $[68,8,32]$, $[77,8,36]$, $[84,8,40]$, $[85,8,40]$, $[92,8,44]$, $[99,8,48]$, $[108,8,52]$, $[115,8,56]$ and $[162,8,80]$ codes. By puncturing the binary SO $[45,8,20]$ code on $\{10,11\}$, we can construct a binary SO $[43,8,18]$ code. By puncturing the binary SO $[62,8,28]$ code on $\{9,27\}$, we can construct a binary SO $[60,8,26]$ code. By puncturing the binary SO $[68,8,32]$ code on $\{64,65\}$, we can construct a binary SO $[66,8,30]$ code. By puncturing the binary SO $[77,8,36]$ code on $\{61,77\}$, we can construct a binary SO $[75,8,34]$ code. By puncturing the binary SO $[85,8,40]$ code on $\{35,52\}$, we can construct a binary SO $[83,8,38]$ code. By puncturing the binary SO $[108,8,52]$ code on $\{4,5\}$, we can construct a binary SO $[106,8,50]$ code. By puncturing the binary SO $[162,8,80]$ code on $\{1,10\}$, we can construct a binary SO $[160,8,78]$ code.
\end{example}

\begin{remark}
We refer to an important upper bound of $d_{so}(n,k)$, namely, $d_{so}(n,k)\leq \left\lfloor\frac{d(n,k)}{2}\right\rfloor$. The upper bound of $d(n,k)$ can be seen in \cite{codetables}. By Corollary \ref{cor-SO-SS}, there is a binary SO $[g(8,D),8,D]$ Griesmer code for $D=64,96,108,112,116,120,124$.
Applying Lemma \ref{lem-g} to these binary SO Griesmer codes, we can partially determine the exact value of $d_{so}(n,8)$ for $41\leq n\leq 255$. Combining Propositions \ref{[253-8]} and Example \ref{example-SO-2}, we give Table 8.
\end{remark}

\begin{center}
{\small
\begin{tabular}{ll|ll}
\multicolumn{4}{c}{{\rm Table 8: Some binary optimal $[N,8]$ SO codes}}\\
\hline
\makebox[0.1\textwidth][l]{$N$}& \makebox[0.01\textwidth][l]{$d_{so}$}&
\makebox[0.1\textwidth][l]{$N$}& \makebox[0.01\textwidth][l]{$d_{so}$}\\
\hline
$41$ & 16  &$153$ & 72-74 \\
$42$ & 16-18 &$154$ & 74 \\

$43,44$ & 18 &$155,156,157,158, 159$ & 76 \\

$45,46,47, 48$ & 20 &$160,161$ & 78 \\

$49,50$ & 20-22 &$162,163,164,165,166,167, 168$ & 80 \\

$51,52,53,54,55,56,57$ & 24 &$169$ & 80-82 \\

$58,59,60$ & 24-26 &$170$ & 82 \\

$60$ & 26 &$171,172,173,174, 175$ & 84 \\

$61,62,63,64$ & 28 &$176,177$ & 86 \\

$65$ & 28-30  &$178,179,180,181,182$ & 88 \\

$66,67$ & 30  &$183,184$ & 88-90 \\

$68,69,70,71,72,73$ & 32 &$185$ & 90 \\

$74$ & 32-34 &$186$ & 90-92\\

$75,76$ & 34 &$187,188, 189, 190$ & 92\\

$77,78,79,80$ & 36  &$191$ & 92-94\\

$81,82$ & 36-38 &$192,193,194,195,196,197,198, 199$ & 96\\

$83$ & 38 &$200$ & 96-98\\

$84,85,86,87,88$ & 40  &$201,202$ & 98\\

$89,90,91$ & 40-42 &$203,204,205,206$ & 100\\

$92,93,94,95$ & 44 &$207$ & 100-102\\

$96,97,98$ & 44-46 &$208,209$ & 102\\

$99,100,101,102,103,104, 105$ & 48 &$210,211,212,213, 214$ & 104\\

$106,107$ & 50 &$215$ & 104-106\\

$108,109,110,111, 112$ & 52 &$216$ & 106\\

$113,114$ & 52-54 &$217,218,219,220, 221,222$ & 108\\

$115,116,117,118,119$ & 56 &$223$ & 110\\

$120,121,123$ & 56-58 &$224,225,226,227,228,229, 230$ & 112\\

$124,125$ & 56-60 &$231$ & 112-114\\

$126,127$ & 56-62 &$232$ & 114\\

$128,129,130,131,132,133,134,135, 136$ & 64 &$233,234,235,236, 237, 238$ & 116\\

$137,138,139$ & 64-66 &$239$ & 118\\

$140,141$ & 64-68 &$240,241,242,243,244, 245, 246$ & 120\\

$142,143, 144$ & 68 &$247$ & 122\\

$145$ & 68-70 &$248,249,250,251,252,253$ & 124\\

$146$ & 70 &$254$ & 126\\

$147,148,149,150,151,152$ & 72 \\
\hline
\end{tabular}}
\end{center}

\section{Conclusion}

In this paper, we have pushed further the study of characterization of optimal binary SO codes. First, binary SO Griesmer codes have been constructed from the Solomon-Stiffler codes. Next, we have proposed a general construction method for binary SO codes. As a consequence, the exact value of $d_{so}(n,k)$ have been determined when $n$ is large relative to $k$. An open problem proposed by Kim and Choi \cite{Kim-SO} has been also pushed greatly. More specifically, we have reduced the problem with an infinite number of cases to the problem of a finite number of cases.
Finally, we have proved the nonexistence of some binary SO codes with dimension 7 by employing the residual code approach.
Our results provide a general method to prove the nonexistence of some binary SO codes.

As future work, the authors want other researchers to tackle the remaining five cases where $k=7,8$. It will be also very difficult to characterize optimal binary SO codes with length $n\geq30$ and dimension $k\geq 9$.\\

\noindent{\bf Conflict of Interest}
The authors have no conflicts of interest to declare that are relevant to the content of this
article.\\

\noindent{\bf Data Deposition Information}
Our data can be obtained from the authors upon reasonable request.\\

\noindent{\bf Acknowledgement}
This research is {\mbox{supp}}orted by Natural Science Foundation of China (12071001).

\end{document}